\documentclass{llncs}
\usepackage[latin2]{inputenc}
\usepackage[english]{babel}
\usepackage[numbers]{natbib}
\usepackage{dsfont, complexity}
\usepackage{amsmath}
\usepackage{amssymb}
\usepackage{float}
\usepackage{graphicx}
\usepackage{tikz}
\usepackage{xcolor}
\usepackage{xspace}
\usepackage{paralist,bold-extra}
\usetikzlibrary{arrows,decorations.markings,patterns,calc, positioning, matrix}
\usepackage{hyperref}
\usepackage[capitalise]{cleveref}

\newtheorem{pr}[theorem]{Problem}
\newtheorem{ex}[theorem]{Example}
\usepackage{algpseudocode}
\usepackage[ruled]{algorithm}

\usepackage{thmtools, mathtools}
\usepackage{thm-restate}

\makeatletter % for pretty appendix titles 
\def\@seccntformat#1{\@ifundefined{#1@cntformat}%
   {\csname the#1\endcsname\quad}  % default
   {\csname #1@cntformat\endcsname}% enable individual control
}
\let\oldappendix\appendix %% save current definition of \appendix
\renewcommand\appendix{%
    \oldappendix
    \newcommand{\section@cntformat}{\appendixname~\thesection\quad}
}
\makeatother

\newcommand{\OPT}{\textup{OPT}}
\newcommand{\tw}{\operatorname{tw}}
\newcommand{\set}[1]{\left\{#1\right\}}

\begin{document}

  \title{Matchings with lower quotas:\\ Algorithms and complexity\thanks{A preliminary version of this paper appeared at the 26th International Symposium on Algorithms and Computation (ISAAC 2015). The authors were supported by the Deutsche Telekom Stiftung, COST Action IC1205 on Computational Social Choice, DFG within project A07 of CRC TRR 154, EPSRC grant EP/K010042/1, DAAD with funds of BMBF and the EU Marie Curie Actions. Part of this work was carried out whilst \'A.~Cseh was visiting the University of Glasgow.}}

\author{
Ashwin Arulselvan\inst{1}\and 
\'{A}gnes Cseh\inst{2}
\and 
Martin Gro{\ss}\inst{3}
\and 
David F. Manlove\inst{4}
\and 
Jannik Matuschke\inst{5}
}

\institute{Department of Management Science, University of Strathclyde, UK\\              \email{ashwin.arulselvan@strath.ac.uk}
           \and
           School of Computer Science, Reykjavik University, Iceland\\ 
           \email{cseh@ru.is}
           \and Institute for Mathematics, TU Berlin, Germany\\
              \email{gross@math.tu-berlin.de}
              \and
           School of Computing Science, University of Glasgow, UK \\ 
           \email{David.Manlove@glasgow.ac.uk}
           \and TUM School of Management, Technische Universit\"{a}t M\"{u}nchen, Germany\\
              \email{jannik.matuschke@tum.de}
}
\date{}
\maketitle

\begin{abstract}
We study a natural generalization of the maximum weight many-to-one matching problem. We are given an undirected bipartite graph $G= (A\, \dot\cup\, P, E)$ with weights on the edges in $E$, and with lower and upper quotas on the vertices in~$P$. We seek a maximum weight many-to-one matching satisfying two sets of constraints: vertices in $A$ are incident to at most one matching edge, while vertices in $P$ are either unmatched or they are incident to a number of matching edges between their lower and upper quota. This problem, which we call \emph{maximum weight many-to-one matching with lower and upper quotas} ({\sc wmlq}), has applications to the assignment of students to projects within university courses, where there are constraints on the minimum and maximum numbers of students that must be assigned to each project.

In this paper, we provide a comprehensive analysis of the complexity of {\sc wmlq} from the viewpoints of classical polynomial time algorithms, fixed-parameter tractability, as well as approximability. We draw the line between $\NP$-hard and polynomially tractable instances in terms of degree and quota constraints and provide efficient algorithms to solve the tractable ones. We further show that the problem can be solved in polynomial time for instances with bounded treewidth; however, the corresponding runtime is exponential in the treewidth with the maximum upper quota $u_{\max}$ as basis, and we prove that this dependence is necessary unless $\FPT = \W[1]$. The approximability of {\sc wmlq} is also discussed: we present an approximation algorithm for the general case with performance guarantee $u_{\max}+1$, which is asymptotically best possible unless $\P = \NP$. Finally, we elaborate on how most of our positive results carry over to matchings in arbitrary graphs with lower quotas.

%\keywords{maximum matching \and many-to-one matching \and project allocation \and inapproximability \and bounded treewidth}
\end{abstract}

\section{Introduction}

Many university courses involve some element of team-based project work. A set of projects is available for a course and each student submits a subset of projects as acceptable. For each acceptable student--project pair $(s,p)$, there is a weight $w(s,p)$ denoting the \emph{utility} of assigning $s$ to~$p$. The question of whether a given project can run is often contingent on the number of students assigned to it. 
Such quota constraints also arise in various other contexts involving the centralized formation of groups, including organizing team-based activities at a leisure center, opening facilities to serve a community and coordinating rides within car-sharing systems. In these and similar applications, the goal is to maximize the utility of the assigned agents under the assumption that the number of participants for each open activity is within the activity's prescribed limits.

We model this problem using a weighted bipartite graph $G= (A\,\dot\cup\,P, E)$, where the vertices in $A$ represent \emph{applicants}, while the vertices in $P$ are \emph{posts} they are applying to. So in the above student--project allocation example, $A$ and $P$ represent the students and projects respectively, and $E$ represents the set of acceptable student--project pairs. The edge weights capture the cardinal utilities of an assigned applicant--post pair. Each post has a lower and an upper quota on the number of applicants to be assigned to it, while each applicant can be assigned to at most one post.  In a feasible assignment, a post is either \emph{open} or \emph{closed}: the number of applicants assigned to an open post must lie between its lower and upper quota, whilst a closed post has no assigned applicant. The objective is to find a maximum weight many-to-one matching satisfying all lower and upper quotas. We denote this problem by {\sc wmlq}.

In this paper, we study the computational complexity of {\sc wmlq} from various perspectives.  We begin by defining the problem formally in \cref{sec:probdef}.  Then in \cref{sec:com_rest}, we show that {\sc wmlq} can be solved efficiently if the degree of every post is at most $2$, whereas the problem becomes hard as soon as posts with degree~$3$ are permitted, even when lower and upper quotas are all equal to the degree, and every applicant has a degree of~$2$. Furthermore, we show the tractability of the case of pair projects, i.e., when all upper quotas are at most~2. In \cref{sec:bounded_tw}, we study the fixed parameter tractability of {\sc wmlq}. To this end, we generalize the known dynamic program for maximum independent set with bounded treewidth to \textsc{wmlq}. The running time of our algorithm is exponential in the treewidth of the graph, with $u_{\max}$, the maximum upper quota of any vertex, as the basis. This yields a fixed-parameter algorithm when parameterized by both the treewidth and $u_{\max}$. We show that this exponential dependence on the treewidth cannot be completely separated from the remaining input by establishing a $W[1]$-hardness result for {\sc wmlq} parameterized by treewidth. Finally, in \cref{se:approx}, we discuss the approximability of the problem. We show that a simple greedy algorithm yields an approximation guarantee of $u_{\max}+1$ for {\sc wmlq} and $\sqrt{|A|}+1$ in the case of unit edge weights. We complement these results by showing that these approximation factors are asymptotically best possible, unless $\P = \NP$. We briefly comment on the generalizability our aforementioned results in \cref{sec:many} for matchings in arbitrary graphs with lower quotas.

\subsection{Related work}\label{sec:rWork}

Among various applications of centralized group formation, perhaps the assignment of medical students to hospitals has received the most attention. In this context, as well as others, the underlying model is a bipartite matching problem involving lower and upper quotas. The \emph{Hospitals / Residents problem with Lower Quotas} ({\sc hrlq})~\cite{BFIM10,HIM14} is a variant of {\sc wmlq} where applicants and posts have ordinal preferences over one another, and we seek a \emph{stable matching} of residents to hospitals. Hamada et al.~\cite{HIM14} considered a version of {\sc hrlq} where hospitals cannot be closed, whereas the model of Bir\'o et al.~\cite{BFIM10} permitted hospital closures. Strategyproof mechanisms have also been studied in instances with ordinal preferences and no hospital closures~\cite{FITUY15,GHIKUYY14,GKHIY15}.

The \emph{Student / Project Allocation problem}~\cite[Section 5.6]{Man13} models the assignment of students to projects offered by lecturers subject to upper and lower quota restrictions on projects and lecturers. Several previous papers have considered the case of ordinal preferences involving students and lecturers~\cite{AIM07,IMY12,MO08} but without allowing lower quotas. However two recent papers~\cite{Kam13,MT13} do permit lower quotas together with project closures, both in the absence of lecturer preferences. Monte and Tumennasan~\cite{MT13} considered the case where each student finds every project acceptable, and showed how to modify the classical ``serial dictatorship'' mechanism to find a Pareto optimal matching. Kamiyama~\cite{Kam13} generalized this mechanism to the case where students need not find all projects acceptable, and where there may be additional restrictions on the sets of students that can be matched to certain projects. This paper also permits lower quotas and project closures, but our focus is on cardinal utilities rather than  ordinal preferences.

The unit-weight version of {\sc wmlq} is also closely related to the \emph{$D$-matching problem}~\cite{Cor88,Lov73,Seb93}, a variant of graph factor problems~\cite{Plu07}.  In an instance of the $D$-matching problem, we are given a graph $G$, and a domain of integers is assigned to each vertex. The goal is to find a subgraph $G'$ of $G$ such that every vertex has a degree in $G'$ that is contained in its domain. Lov\'asz~\cite{Lov72} showed that the problem of deciding whether such a subgraph exists is $\NP$-complete, even if each domain is either $\{1\}$ or~$\{0,3\}$. On the other hand, some cases are tractable.  For example, if for each domain $D$, the complement of $D$ contains no consecutive integers, the problem is polynomially solvable~\cite{Seb93}. 
As observed in~\cite{SS11}, $D$-matchings are closely related to \emph{extended global cardinality constraints} and the authors provided an analysis of the fixed-parameter tractability of a special case of the $D$-matching problem; see~\cref{sec:bounded_tw} for details.

The problem that we study in this paper corresponds to an optimization version of the $D$-matching problem. We consider the special case where $G$ is bipartite and the domain of each applicant vertex is $\{0,1\}$, whilst the domain of each post vertex $p$ is $\{0\}\cup \{\ell(p),\dots,u(p)\}$, where $\ell(p)$ and $u(p)$ denote the lower and upper quotas of $p$ respectively.  Since the empty matching is always feasible in our case, our aim is to find a domain-compatible subgraph $G'$ such that the total weight of the edges in $G'$ is maximum.

\section{Problem definition}
\label{sec:probdef}
In this section we provide a formal definition of the maximum weight many-to-one matching problem with lower quotas (\textsc{wmlq}).

\paragraph{Basic notation}
Let $G = (V, E)$ be a graph.
For a subset of vertices $U \subseteq V$ we denote by $\delta(U) = \{\{ v, w\} \in E : v \in U, w\in V \setminus U\}$ the set of edges incident to exactly one vertex in $U$. For a vertex $v \in V$, we write $\delta(v) = \delta(\{v\})$, and for a subset of edges $F \subseteq E$ we write $\deg_F(v) = |\delta(v) \cap F|$.
By $\Gamma(v) = \{w \in V : \{v, w\} \in E\}$ we denote the \emph{neighborhood} of $v$, i.e., the set of vertices that are adjacent to~$v$.
\medskip

In our problem, a set of applicants $A$ and a set of posts $P$ are given. $A$ and $P$ constitute the two vertex sets of an undirected bipartite graph~$G = (V, E)$ with $V = A\,\dot\cup\, P$ and $E$ represents the set of acceptable applicant-post pairs. 
Each edge carries a \emph{weight} $w: E \rightarrow \mathbb{R}_{\geq 0}$, representing the utility of the corresponding assignment. The set of posts is equipped with functions $\ell: P \rightarrow \mathbb{Z}_{\geq 0}$ and $u: P \rightarrow \mathbb{Z}_{\geq 0}$ such that $\ell(p) \leq u(p)$ for every $p \in P$.  Here $\ell(p)$ is called the \emph{lower quota} of $p$ and $u(p)$ is called the \emph{upper quota} of $p$.  These functions bound the number of admissible applicants for the post (independent of the weight of the corresponding edges). Furthermore, every applicant can be assigned to at most one post.
Thus, an \emph{assignment} is a subset $M \subseteq E$ of the edges such that $|\delta(a) \cap M| \leq 1$ for every applicant $a \in A$ and $|\delta(p) \cap M| \in \left\{ 0, \ell(p), \ell(p)+1, ..., u(p) \right\}$ for every $p \in P$. With respect to an assignment $M$, a post is said to be \emph{open} if the number of applicants assigned to it is greater than~$0$, and \emph{closed} otherwise. The \emph{size} of an assignment $M$, denoted $|M|$, is the number of assigned applicants, while the \emph{weight} of $M$, denoted $w(M)$, is the total weight of the edges in $M$, i.e., $w(M) = \sum_{e \in M} w(e)$. The goal is to find an assignment of maximum weight.
 
\begin{remark}
Note that when not allowing closed posts, the problem immediately becomes tractable. It is easy to see this in the unweighted case as any algorithm for maximum flow with lower capacities can be used to determine an optimal solution in polynomial time. Maximum flow with lower capacities can be easily reduced to the classical maximum flow problem. The method can be naturally extended to the weighted case as the flow-based linear program has integral extreme points due to its total unimodularity property.
\end{remark}

\begin{pr}\textsc{wmlq}\ \\
	Input: $\mathcal{I} = (G, w, \ell, u)$; a bipartite graph $G = (A\,\dot\cup\,P, E)$ with edge weights $w$, lower quotas $\ell$ and upper quotas $u$.\\
Task: Find an assignment of maximum weight.\\
If $w(e)=1$ for all $e \in E$, we refer to the problem as \textsc{mlq}.
\label{pr:wmlq}
\end{pr}
 
%We denote the number of vertices $|A| + |B|$ by $n$, while $|E| = m$.
 
Some trivial simplification of the instance can be executed right at the start. If $u(p) > |\Gamma(p)|$ for a post $p$, then $u(p)$ can be replaced by $|\Gamma(p)|$. On the other hand, if $\ell(p) > |\Gamma(p)|$, then post $p$ can immediately be deleted, since no feasible solution can satisfy the lower quota condition. Moreover, a post $p$ with $\ell(p) = 1$ behaves identically to the case that $\ell(p) = 0$, so we assume that no post $p$ has $\ell(p) = 1$. From now on we assume that the instances have already been simplified this way.

\section{Degree- and quota-restricted cases}\label{sec:com_rest}
In this section we characterize the complexity of {\sc wmlq} in the presence of upper bounds placed on vertex degrees or the posts' upper quotas.  \cref{sec:degreerest} deals with degree-restricted cases, whilst \cref{sec:quotarest} studies cases involving bounded upper quotas.

\subsection{Degree-restricted cases}
\label{sec:degreerest}
In this subsection we will consider \textsc{wmlq}$(i,j)$, the special case of \textsc{wmlq} in which $|\Gamma(a)|\leq i$ for all $a\in A$, and $|\Gamma(p)|\leq j$ for all $p\in P$. That is, every applicant submits at most $i$ applications and every post receives at most $j$ applications. In order to establish our first result, we reduce the maximum independent set problem (\textsc{mis}) to \textsc{mlq}. In \textsc{mis}, a graph with $n$ vertices and $m$ edges is given and the task is to find an independent vertex set of maximum size. 
\textsc{mis} is not approximable within a factor of~$n^{1-\varepsilon}$ for any~$\varepsilon > 0$, unless $\P = \NP$~\cite{Zuc07}. The problem remains $\APX$-complete even for cubic (3-regular) graphs~\cite{AK00}.

\begin{theorem}
	\label{th:max_spa_np}
	\textsc{mlq(2,3)} is $\APX$-complete.
\end{theorem}

\begin{proof}
	First of all, \textsc{mlq(2,3)} is in $\APX$ because %feasible solutions are of polynomial size and 
    the problem has a $4$-approx\-ima\-tion that can be found in polynomial time (see Theorem~\ref{thm:greedy-approximation}).
	
	To each instance $\mathcal{I}$ of \textsc{mis} on cubic graphs we create an instance $\mathcal{I}'$ of \textsc{mlq} such that there is an independent vertex set of size at least $K$ in $\mathcal{I}$ if and only if $\mathcal{I}'$ admits an assignment of size at least~$3K$, yielding an approximation-preserving reduction.  The construction is as follows. To each of the $n$ vertices of graph $G$ in~$\mathcal{I}$, a post with upper and lower quota of~3 is created. The $m$ edges of $G$ are represented as $m$ applicants in~$\mathcal{I}'$. For each applicant $a \in A$, $|\Gamma(a)| =2$ and $\Gamma(a)$ comprises the two posts representing the two end vertices of the corresponding edge. Since we work on cubic graphs, $|\Gamma(p)| = 3$ for every post~$p \in P$.
	
	First we show that an independent vertex set of size $K$ can be transformed into an assignment of at least $3K$ applicants. All we need to do is to open a post with its entire neighborhood assigned to it if and only if the vertex representing that post is in the independent set. Since no two posts stand for adjacent vertices in~$G$, their neighborhoods do not intersect. Moreover, the assignment assigns exactly three applicants to each of the $K$ open posts.
	
	To establish the opposite direction, let us assume that an assignment of cardinality at least $3K$ is given. The posts' upper and lower quota are both set to~3, therefore, the assignment involves at least $K$ open posts. No two of them can represent adjacent vertices in $G$, because then the applicant standing for the edge connecting them would be assigned to both posts at the same time.
    
    %The reduction given here is an $L$-reduction~\cite{PY91} with constants $\alpha=\beta=3$. 
		Note that every solution of the constructed instance of \textsc{mlq} serves an integer multiple of 3 applicants. In particular, the \textsc{mlq} instance has a solution serving $3K$ applicants if and only if there is an independent set of size $K$ in the \textsc{mis} instance. Hence, this reduction  preserves the approximation factors. Since \textsc{mlq(2,3)} belongs to $\APX$ and \textsc{mis} is $\APX$-complete in cubic graphs, it follows that \textsc{mlq(2,3)} is $\APX$-complete.
\qed
\end{proof}

So far we have established that if $|\Gamma(a)| \leq 2$ for every applicant $a \in A$ and $|\Gamma(p)| \leq 3$ for every post $p \in P$, then \textsc{mlq} is $\NP$-hard. In the following, we also show that these restrictions are the tightest possible. If $|\Gamma(p)| \leq 2$ for every post $p \in P$, then a maximum weight matching can be found efficiently, regardless of~$|\Gamma(a)|$. Note that the case \textsc{wmlq(1,$\infty$)} is trivially solvable.

\begin{theorem}
\label{th:infty_2}
	\textsc{wmlq($\infty$,2)} is solvable in $O(n^2 \log n)$ time, where $n = |A| + |P|$.
\end{theorem}

\begin{proof}
	After executing the simplification steps described at the end of \cref{sec:probdef}, we apply two more changes to derive our helper graph~$H$. Firstly, if $\ell(p) = 0$, $u(p) = 2$ and $|\Gamma(p)| = 2$, we separate $p$'s two edges, splitting $p$ into two posts with upper quota~1. After this step, all posts with $u(p) = 2$ also have $\ell(p) = 2$. All remaining vertices are of upper quota~1. Then, we substitute all edge pairs of posts with $\ell(p) = u(p) = 2$ with a single edge connecting the two applicants. This edge will carry the weight equal to the sum of the weights of the two deleted edges.
    
    Clearly, any matching in $H$ translates into an assignment of the same weight in $G$ and vice versa. Finding a maximum weight matching in a general graph with $n$ vertices and $m$ edges can be done in $O(n(m + n \log n))$ time~\cite{Gab90}, which reduces to $O(n^2 \log n)$ in our case. \qed
%
    %From this graph we construct a helper graph $H$ with edge weights and show that the weight of any matching in $H$ corresponds to the weight of an assignment in~$G$. Finding a maximum weight matching in a general graph with $n$ vertices and $m$ edges can be done in $O(n(m + n \log n))$ time~\cite{Gab90}, which reduces to $O(n^2 \log n)$ in our case.
%
    %To construct $H$, we start with the simplified $G$ and substitute all edge pairs of posts with $\ell(p) = u(p) = 2$ with a single edge connecting the two applicants. This edge will carry the weight equal to the sum of the weights of the two deleted edges. Clearly, any matching in $H$ translates into an assignment of the same weight in $G$ and vice versa.
\end{proof}

\subsection{Quota-restricted cases}
\label{sec:quotarest}
In this section, we consider restrictions of \textsc{wmlq} with bounded upper quotas. Note that Theorem~\ref{th:max_spa_np} already tells us that the case of $u(p) \leq 3$ for all posts $p \in P$ is $\NP$-hard to solve. We will now settle the complexity of the only remaining case, where we have instances with every post $p\in P$ having an arbitrary degree and $u(p) \le 2$. This setting models posts that need to be assigned  to none, one or pairs of applicants.

%The problem is connected to various known problems in graph theory, one of them being the \emph{$S$-path packing problem}. In that problem, we are given a graph with a set of terminal vertices~$S$. The task is to pack the highest number of vertex-disjoint paths so that each path starts and ends at a terminal vertex, while all its inner vertices are non-terminal. The problem can be solved in $O(n^{2.38})$ time~\cite{CLL14,SS04} with the help of matroid matching~\cite{Lov80}. An instance of \textsc{mlq} with $\ell(p) = u(p) = 2$ for every post $p \in P$ corresponds to an $S$-path packing instance with $S = A$. The highest number of vertex-disjoint paths starting and ending in $A$ equals half of the cardinality of a maximum assignment. Thus, \textsc{mlq} with $\ell(p) = u(p) = 2$ for all $p\in P$ can also be solved in $O(n^{2.38})$ time. On the other hand, there is no straightforward way to model posts with $u(p) = 1$ in $S$-path packing and introducing weights to the instances also seems to be a challenging task. Some progress has been made for weighted edge-disjoint paths, but to the best of our knowledge the question is unsettled for vertex-disjoint paths~\cite{HP14}.

Here we present a solution for \textsc{wmlq} with $u(p) \leq 2$. Our algorithm is based on $f$-factors of graphs. In the {\em $f$-factor problem}, a graph $G$ and a function $f: V \rightarrow \mathbb{Z}_{\geq 0}$ is given. A set of edges $F \subseteq E$ is called an \emph{$f$-factor} if $\deg_F(v) = f(v)$ for every $v \in V$, where $\deg_F(v)$, as per our earlier definition, is the degree of $v$ in the graph~$(V,F)$. Constructing an $f$-factor of maximum weight in a graph with $n$ vertices and $m$ edges or proving that none exists can be done in $O(\phi(m + n \log{n}))$ time, where $\phi$ is the sum of all $f$-values in the graph~\cite{Gab83,Gab90}.

\begin{restatable}{theorem}{restatepairs}
\label{th:u2}
	\textsc{wmlq} with $u(p) \leq 2$ for every $p \in P$ can be solved in $O(nm + n^2 \log{ n})$ time, where $n = |V|$ and $m = |E|$.
\end{restatable}

\begin{proof}

We partition $P$ into $P_1$ and $P \setminus P_1$, where $P_1$ denotes the set of posts with $u(p) = 1$. For posts in $P \setminus P_1$ we can assume that $\ell(p) = u(p) = 2$ for every post~$p$. For, a post $p$ with $\ell(p)=0$ and $u(p) = 2$ can be transformed into a post with $\ell(p) = u(p) = 2$ by giving it two dummy edges with zero weight, allowing us to pick the dummy edges in order to make up for the raised lower quota. %A post $p$ with $\ell(p)=0$ and $u(p) = 1$ will be transformed into a post with $\ell(p) = u(p) = 1$ without altering its behavior.  Let us denote the set of posts with $\ell(p) = u(p) = 1$ by~$P_1$.

The graph $G' =(V',E')$ of the constructed $f$-factor instance contains the graph $G =(V,E)$ of our \textsc{wmlq} instance, as shown in \cref{fi:ffactor}. We add a dummy post $p_d$ to $V'$ and connect it to every applicant in~$A$. We connect every post $p_i\in P_1$ to~$p_d$. For every post $p_i\in P\backslash P_1$ we add two dummy vertices $q_i^1$ and $q_i^2$ and a triangle on the vertices $p_i, q_i^1$ and~$q_i^2$. All new edges in $E' \setminus E$ carry zero weight. %There are $|A| + |P_1|$ triangles added to $p_d$ as well. 
 
We set $f(p_d) = K$, $f(p) = u(p)$ for every $p \in P$ and $f(v) = 1$ for the rest of the vertices. In the initial version of our algorithm, we solve a weighted $f$-factor problem for every $K \in \{0, 1, ..., |A| + |P_1|\}$, and later we will show a slightly modified version of the $f$-factor instance so that it is sufficient to construct only two instances.
    
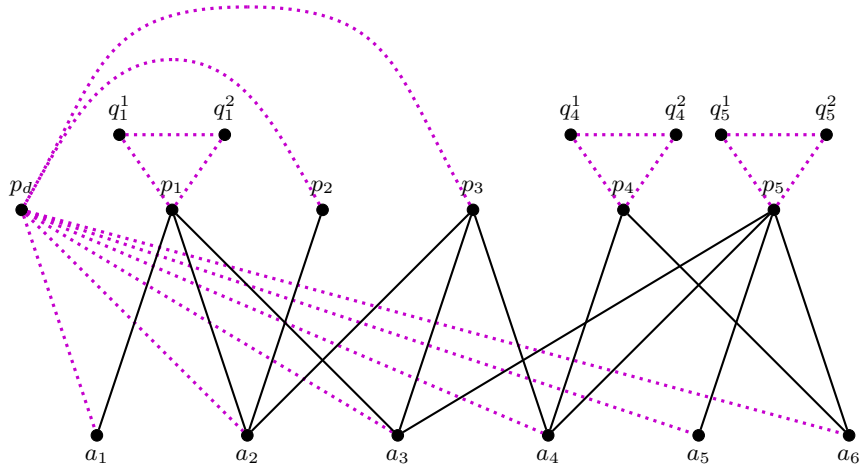
\begin{figure}[ht]
\centering
\begin{tikzpicture}[scale=1, transform shape]

\tikzstyle{vertex} = [circle, draw=black, fill=black, scale= 0.5]
\tikzstyle{edgelabel} = [rectangle, fill=white]
\definecolor{MyPurple}{RGB}{197,0,205}
\pgfmathsetmacro{\d}{0.7}
\pgfmathsetmacro{\b}{3}

%%%%%%%%%%%%%%%%%%%%%%%%%%%%%%%%%%%%%%%%%%%%%%%%%%%

\node[vertex, label=above:$p_d$] (m_d) at (-2, 3) {};
\node[vertex, label=above:$p_1$] (m_1) at (0, 3) {};
\node[vertex, label=above:$p_2$] (m_2) at (2, 3) {};
\node[vertex, label=above:$p_3$] (m_3) at (4, 3) {};
\node[vertex, label=above:$p_4$] (m_4) at (6, 3) {};
\node[vertex, label=above:$p_5$] (m_5) at (8, 3) {};

\node[vertex, label=below:$a_1$] (w_1) at (-1, 0) {};
\node[vertex, label=below:$a_2$] (w_2) at (1, 0) {};
\node[vertex, label=below:$a_3$] (w_3) at (3, 0) {};
\node[vertex, label=below:$a_4$] (w_4) at (5, 0) {};
\node[vertex, label=below:$a_5$] (w_5) at (7, 0) {};
\node[vertex, label=below:$a_6$] (w_6) at (9, 0) {};

\node[vertex, label=above:$q_1^1$] (q_11) at ($(m_1) + (-\d, 1)$) {};
\node[vertex, label=above:$q_1^2$] (q_12) at ($(m_1) + (\d, 1)$) {};

\node[vertex, label=above:$q_4^1$] (q_41) at ($(m_4) + (-\d, 1)$) {};
\node[vertex, label=above:$q_4^2$] (q_42) at ($(m_4) + (\d, 1)$) {};

\node[vertex, label=above:$q_5^1$] (q_51) at ($(m_5) + (-\d, 1)$) {};
\node[vertex, label=above:$q_5^2$] (q_52) at ($(m_5) + (\d, 1)$) {};

\node[vertex, shape=coordinate] (q_21) at ($(m_2) + (-2, 2)$) {};
\node[vertex, shape=coordinate] (q_31) at ($(m_3) + (-3, 2.7)$) {};

\draw [very thick, dotted, MyPurple] (m_d) --  (w_1);
\draw [very thick, dotted, MyPurple] (m_d) --  (w_2);
\draw [very thick, dotted, MyPurple] (m_d) --  (w_3);
\draw [very thick, dotted, MyPurple] (m_d) --  (w_4);
\draw [very thick, dotted, MyPurple] (m_d) --  (w_5);
\draw [very thick, dotted, MyPurple] (m_d) --  (w_6);

\draw [very thick, dotted, MyPurple] (q_41) --  (q_42);
\draw [very thick, dotted, MyPurple] (m_4) --  (q_41);
\draw [very thick, dotted, MyPurple] (m_4) --  (q_42);

\draw [very thick, dotted, MyPurple] (q_11) --  (q_12);
\draw [very thick, dotted, MyPurple] (m_1) --  (q_11);
\draw [very thick, dotted, MyPurple] (m_1) --  (q_12);

\draw [very thick, dotted, MyPurple] (q_51) --  (q_52);
\draw [very thick, dotted, MyPurple] (m_5) --  (q_51);
\draw [very thick, dotted, MyPurple] (m_5) --  (q_52);

\draw [thick] (m_1) -- (w_1);
\draw [thick] (m_1) -- (w_2);
\draw [thick] (m_1) -- (w_3);
\draw [thick] (m_2) -- (w_2);
\draw [thick] (m_3) -- (w_3);
\draw [thick] (m_3) -- (w_2);
\draw [thick] (m_4) -- (w_4);
\draw [thick] (m_4) -- (w_6);
\draw [thick] (m_5) -- (w_4);
\draw [thick] (m_5) -- (w_3);
\draw [thick] (m_3) -- (w_4);
\draw [thick] (m_5) -- (w_5);
\draw [thick] (m_5) -- (w_6);

\draw [very thick, dotted, MyPurple] (m_2) to[out=120,in=0, distance=1cm ] (q_21);
\draw [very thick, dotted, MyPurple] (q_21) to[out=180,in=60, distance=1cm ] (m_d);
\draw [very thick, dotted, MyPurple] (m_3) to[out=120,in=0, distance=2cm ] (q_31);
\draw [very thick, dotted, MyPurple] (q_31) to[out=180,in=60, distance=2cm ] (m_d);

\end{tikzpicture}
\caption{The transformation from \textsc{wmlq} to an $f$-factor problem. The solid edges form $G$, while the dotted edges are the added ones, carrying weight~0.  Here, $P_1=\{p_2,p_3\}$ and $P\backslash P_1=\{p_1,p_4,p_5\}$.}
\label{fi:ffactor}
\end{figure}

First we show that if there is a feasible assignment $M$ in $G$ so that the number of unmatched applicants and the number of closed posts in $P_1$ add up to $K$, then it can be extended to an $f$-factor $M'$ of the same weight in~$G'$. We construct $M'$ starting with $M$ and then adding the following edges to it:

\begin{itemize}
 \item $\set{p_d,a_i}$ for every applicant $a_i$ that is unmatched in~$M$;
 \item $\set{q_i^1,p_i}$ and $\set{q_i^2,p_i}$ for every post $p_i\in P\backslash P_1$ that is closed in~$M$;
 \item $\set{q_i^1,q_i^2}$ for every post $p_i\in P\backslash P_1$ that is open in~$M$;
 \item $\set{p_d,p_i}$ for every post $p_i \in P_1$ that is closed in~$M$;
 %\item $\set{p_d,q_i^1}$ for every post $p_i \in P_1$ that is open in~$M$.
\end{itemize}

For all vertices~$v \neq p_d$, it immediately follows from the construction that $\deg_{M'}(v) = f(v)$. The same holds for $p_d$ as well, because an edge is assigned to it either because an applicant is unmatched or because a post in $P_1$ is closed and we assumed that these add up to~$K$. 

It is easy to see that if there is an $f$-factor $M'$ in $G'$, then its restriction to $G$ is a feasible assignment $M$ of the same weight so that the number of unmatched applicants and the number of closed posts in $P_1$ add up to~$K$. Since every post $p_i \in P_1$ is connected to $p_d$ and $f(p_i) = 1$, it is either the case that $p_i$ is open in $M$ or $\{p_d,p_i\}\in M'$. Regarding posts outside of $P_1$, we need to show that the two edges incident to them are either both in $G$ or neither of them are in~$G$. Assume without loss of generality that $\{p_i,q_i^1\}\in M'$ and $\{p_i,q_i^2\}\notin M'$ for some $p_i \notin P_1$. Since $f(q_i^2) = 1$ and $\deg_{M'}(q_i^2) = 0$, $M'$ cannot be an $f$-factor.

So far we have shown that it is sufficient to test  $|A| + |P_1| +1$ values for $f(p_d)$, and collect the optimal assignments given by the maximum weight $f$-factors. Comparing the weight of these locally optimal solutions delivers a global optimum. A slight modification on the the graph corresponding to the $f$-factor instance will allow us to solve the problem by constructing just two instances, as against $|A| + |P_1| +1$ instances. Similar to the triangles attached to posts in $P \setminus P_1$, triangles are added to $p_d$ as well. The added vertices have $f$-value~1 and the added edges carry weight~0. The number of such triangles hanging on $p_d$ is~$\left\lceil\frac{|A| + |P_1|}{2}\right\rceil$. These triangles can take up all the $f$-value of $p_d$ if necessary, but by choosing the edge not incident to $p_d$ they can also allow $p_d$ to fill up its $f$-value with other edges. Since a triangle either takes up 0 or 2 of $p_d$'s $f$-value, we need to separate the two different parity cases. Thus, to cover all the $|A| + |P_1| + 1$ cases for possible values for $f(p_d)$, in one instance we set $f(p_d)$ to $|A| + |P_1| + 1$ and in the other instance $f(p_d) = |A| + |P_1|$.
\qed
\end{proof}

%-------------------------------------------------

\section{Bounded treewidth graphs}
\label{sec:bounded_tw}

In this section, we investigate \textsc{wmlq} from the point of view of fixed-parameter tractability and analyze how efficiently the problem can be solved for instances with a bounded treewidth.

\paragraph{Fixed-parameter tractability.}
This field of complexity theory is motivated by the fact that in many applications of optimization problems certain input parameters stay small even for large instances. 
A problem, parameterized by a parameter $k$, is fixed-parameter tractable ($\FPT$) if there is an algorithm solving it in time $f(k)\cdot \phi(n)$, where $f : \mathbb{R} \rightarrow \mathbb{R}$ is a function, $\phi$ is a polynomial function, and $n$ is the input size of the instance. Note that this definition not only requires that the problem can be solved in polynomial time for instances where $k$ is bounded by a constant, but also that the dependence of the running time on $k$ is separable from the part depending on the input size. On the other hand, if a problem is shown to be \emph{$\W[1]$-hard}, then the latter property can only be fulfilled if $\FPT = \W[1]$, which would imply $\NP \subseteq \DTIME(2^{o(n)})$. For more details on fixed-parameter algorithms see, e.g.,~\cite{Nie06}.

\paragraph{Treewidth.}
In case of {\sc wmlq} we focus on the parameter \emph{treewidth}, which, on an intuitive level, describes the likeness of a graph to a tree.
A \emph{tree decomposition} of graph $G$ consists of a tree whose nodes---also called \emph{bags}---are subsets of~$V(G)$. These must satisfy the following three requirements.
\begin{enumerate}
\item Every vertex of $G$ belongs to at least one bag of the tree.
\item For every edge $\{a, p\} \in E(G)$, there is a bag containing both $a$ and $p$.
\item If a vertex in $V(G)$ occurs in two bags of the tree, then it also occurs in all bags on the unique tree-path connecting them.
\end{enumerate}
The \emph{width} of a tree decomposition with a set of bags $\mathcal{B}$ is $\max_{B \in \mathcal{B}} |B| - 1$.
The \emph{treewidth} of a graph $G$, $\tw(G)$, is the smallest width among all tree decompositions of~$G$. It is well known that a tree decomposition of smallest width can be found by a fixed-parameter algorithm when parameterized by $\tw(G)$~\cite{Bod96}.
\medskip

In the following, we show that {\sc wmlq} is fixed-parameter tractable when parameterized simultaneously by the treewidth and $u_{\max}$, whereas it remains $W[1]$-hard when only parameterized by the treewidth.
A similar study of the fixed-parameter tractability of the related \emph{extended global cardinality constraint problem} (\textsc{egcc}) was conducted in~\cite{SS11}.
\textsc{egcc} corresponds to the special case of the $D$-matching problem where the graph is bipartite and on one side of the bipartition all vertices have the domain $\{1\}$. In contrast with \textsc{wmlq}, \textsc{egcc} is a feasibility problem (note that the feasibility version of \textsc{wmlq} is trivial, as the empty assignment is always feasible).
 The authors of~\cite{SS11} provided a fixed-parameter algorithm for \textsc{egcc} when parameterized simultaneously by the treewidth of the graph and the maximum domain size, and they showed that the problem is $\W[1]$-hard when only parameterized by the treewidth. These results mirror our results for \textsc{wmlq}, and indeed both our FPT-algorithm for \textsc{wmlq} and the one in~\cite{SS11} are extensions of the same classic dynamic program for the underlying maximum independent set problem. However, our hardness result uses a completely different reduction than the one in~\cite{SS11}. The latter makes heavy use of the fact that the domains can be arbitrary sets, whereas in \textsc{wmlq}, we are confined to intervals.

\subsection{Algorithm for bounded treewidth graphs}

We will now describe an algorithm for solving \textsc{wmlq} in polynomial time for graphs with constant treewidth. The algorithm is a dynamic program that inductively computes a set of partial solutions for each bag of the tree decomposition. We will show how to obtain these solutions for a bag from the solutions of the children of that bag by a sequence of lemmas. Before we state these lemmas, we need to introduce a few more concepts.

\paragraph{Nice tree decomposition.}
For every tree decomposition with a specific treewidth, a \emph{nice tree decomposition}
of the same treewidth can be found in linear time~\cite{Klo94}. A nice tree decomposition
is characterized by an exclusive composition of the following four kinds of bags:
\begin{itemize}
\item leaf bag: $|B| = 1$ and $B$ has no child;
\item introduce bag: $B$ has exactly one child $B_1$, so that $B_1 \subset B$ and $|B \setminus B_1| = 1$;
\item forget bag: $B$ has exactly one child $B_1$, so that $B \subset B_1$ and $|B_1 \setminus B| = 1$;
\item join bag: $B$ has exactly two children $B_1$ and $B_2$, so that $B = B_1 = B_2$.
\end{itemize}
We will henceforth assume we are given such a nice tree decomposition.

%We will take such a nice decomposition of treewidth $tw$ and an arbitrary bag to be the root of the tree. A bag $b = \{a_1, a_2, \dots,  a_{|A_b|}, p_1,\dots, p_{|P_b|}\}$ comprises of a subset of applicants, $A_b = S \cap b$ and a subset of posts $P_b = P \cap b$. Let $V_b \subseteq V(G)$ denote the set of vertices contained in the union of bags present in the subtree rooted at~$b$. Let $G_b$ denote the graph induced by~$V_b$. Note that by this notation, we have $G_{root} = G$. 

\paragraph{Basic notation.}
For ease of exposition, we will define $\ell(a) := u(a) := 1$ for all $a \in A$. Furthermore, throughout this section we will deal with vectors $\alpha \in \mathbb{Z}^U$ for some $U \subseteq V$. We define the notion of extension and restriction of such a vector $\alpha$.
For $U' \subseteq U$ and $\alpha \in \mathbb{Z}^U$ define $\alpha|_{U'}$ as the \emph{restriction} of $\alpha$ to $U'$, i.e., $\alpha|_{U'} \in \mathbb{Z}^{U'}$ and $\alpha|_{U'}(v) = \alpha(v)$ for all $v \in U'$. For $v \in V \setminus U$ and $i \in \mathbb{Z}$ let further $[\alpha, i]_v$ be the \emph{extension} of $\alpha$ to $U \cup \{v\}$ defined by $[\alpha, i]_v(v') := \alpha(v')$ for all $v' \in U$ and $[\alpha, i]_v(v) := i$. 
For a set of edges $S$ we define $\alpha_{S, U} \in \mathbb{Z}^{U}$ by $\alpha_{S, U}(v) := \deg_S(v)$, for all $v\in U$.
We will also use the standard notation $E[S]$ for the set of edges with both endpoints in $S \subseteq V$.

\paragraph{Assignment vectors.}
For any bag $B$, let $V_B \subseteq V$ denote the set of vertices contained in the union of bags present in the subtree rooted at~$B$. We will define a graph $G_B = (V_B, E_B)$ where $E_B := E[V_B] \setminus E[B]$. A \emph{partial assignment} for bag $B$ is an assignment $M \subseteq E_B$ of $G_B$ such that $\deg_M(v) = 0$ or $\ell(v) \leq \deg_M(v) \leq u(v)$ for all $v \in V_B \setminus B$. Note that this definition allows applicants and posts in $B$ to be assigned arbitrarily often and that by definition of $G_B$, no vertex in $B$ is assigned to another vertex in~$B$. An \emph{assignment vector} for bag $B$ is a vector $\alpha \in X_B := \{0, \dots, u_{\max}\}^B$. We say a partial assignment $M$ for $B$ \emph{agrees} with an assignment vector $\alpha \in X_B$, if $\alpha(v) = \deg_M(v)$ for all $v \in B$.
For every bag $B$ and every $\alpha \in X_B$, let $\mathcal{M}_B(\alpha)$ be the set of partial assignments for $B$ that agree with $\alpha$ and let 
\[W_B(\alpha) := \max\left\{\{w(M) : M \in \mathcal{M}_B(\alpha)\right\} \cup \{-\infty\}\}\]
denote the optimal value of any assignment that agrees with $\alpha$ for the graph~$G_B$ (note that a value of $-\infty$ implies that no partial assignment $M$ agrees with~$\alpha$). We further denote the set of optimal partial assignments agreeing with $\alpha$ by
\[\mathcal{M}^*_B(\alpha) := \{M \in \mathcal{M}_B(\alpha) : w(M) = W_B(\alpha)\}.\]
In the following, we will provide a series of lemmas that reveals how to efficiently obtain an element of $\mathcal{M}^*_B(\alpha)$ for every $\alpha \in X_B$ for a bag $B$  (or showing $W_B(\alpha) = -\infty$), assuming such representatives of each set $\mathcal{M}^*_{B'}(\alpha)$ have already been computed for every child $B'$ of $B$ for all $\alpha \in X_{B'}$.

\begin{lemma}\label{lem:leaf}
	Let $B$ be a leaf bag. Then $\mathcal{M}^*_B(0) = \{\emptyset\}$ and $\mathcal{M}^*_B(\alpha) = \emptyset$ for any $\alpha \in X_B \setminus \{0\}$.
\end{lemma}
\begin{proof}
This follows directly from the fact that $E_B = \emptyset$ for all leaf bags and thus the only element in $B$ cannot be assigned.\qed
\end{proof}

\begin{lemma}\label{lem:introduce}
	Let $B$ be an introduce bag such that $B'$ is the only child of $B$ and $B \setminus B' = \{v'\}$. Let $\alpha \in X_B$. Then
	 \begin{align*}
	   \mathcal{M}^*_B(\alpha) = \begin{cases}
	    \mathcal{M}^*_{B'}(\alpha|_{B'}) & \textup{if } \alpha(v') = 0,\\
	    \emptyset & \textup{otherwise.}
	  \end{cases}
	\end{align*}
\end{lemma}
\begin{proof}
  Note that $\Gamma(v') \cap V_B \subseteq B$ by Properties~2 and 3 of a tree decomposition. This implies $\delta(v') \cap E_B = \emptyset$ and hence the lemma follows.\qed
\end{proof}

\begin{lemma}\label{lem:forget}
	Let $B$ be a forget bag such that $B'$ is the unique child of $B$ and $B = B' \setminus \{v'\}$ for some $v' \in B'$. Let $\alpha \in X_B$. 
	Let $(S^*, i^*)$ be an optimal solution to
	\begin{align*}
	  \textup{[forget]}\qquad\max \ \ & w(S) + W_{B'}([\alpha - \alpha_{S, B},\, i - |S|]_{v'})\\
	  \textup{s.t.} \ \ & |S| \leq i,\\
	  & \alpha_{S, B} \leq \alpha\\
	  & S \subseteq \delta(v') \cap \delta(B),\\ 
	  & i \in \{0, \ell(v'), \dots, u(v')\}.
	\end{align*}
	Then $M \cup S^* \in \mathcal{M}^*_B(\alpha)$ for all $M \in \mathcal{M}_{B'}^*([\alpha - \alpha_{S, B}, i - |S|]_{v'})$. If the optimal solution to [forget] has value $-\infty$, then $\mathcal{M}^*_B(\alpha) = \emptyset$.
\end{lemma}
\begin{proof}
  Assume $\mathcal{M}_B(\alpha) \neq \emptyset$ and let $M' \in \mathcal{M}^*_B(\alpha)$. Let \mbox{$S' := M' \cap \delta(v') \cap \delta(B)$} and let $i' := \deg_{M'}(v')$. Observe that $(S', i')$ is a feasible solution to [forget] and that $M' \setminus S' \in \mathcal{M}_{B'}([\alpha - \alpha_{S',B},\, i' - |S'|]_{v'})$. We conclude that $w(M') \leq w(S') + W_{B'}([\alpha - \alpha_{S',B},], i' - |S'|]_{v'}) \leq w(S^*) + W_{B'}([\alpha - \alpha_{S^*,B},\, i^* - |S^*|]_{v'})$. In particular, this implies that the optimal solution value of [forget] is finite and thus there is some \mbox{$M \in \mathcal{M}^*_{B'}([\alpha - \alpha_{S^*,B},\, i^* - |S^*|]_{v'})$}.

  Thus let $M^* := M \cup S^*$. Observe that indeed $\deg_{M^*}(v) = \deg_M(v) + \deg_{S^*}(v) = \alpha(v) - \alpha_{S^*, B}(v) + \alpha_{S^*, B}(v) = \alpha(v)$ for all $v \in B$. Furthermore $\deg_{M^*}(v) = \deg_M(v) \in \{0, \ell(v), \dots, u(v)\}$ for all $v \in V_B \setminus B'$ by feasibility of $M$. Finally, $\deg_{M^*}(v') = i^* \in \{0, \ell(v'), \dots, u(v')\}$, implying $M^* \in \mathcal{M}_B(\alpha)$. As $w(M^*) = w(S^*) + W_{B'}([\alpha - \alpha_{S^*},\, i^* - |S^*|]_{v'}) \geq w(M')$, we conclude that indeed $M^* \in \mathcal{M}^*_B(\alpha)$.
   \qed
\end{proof}

\begin{lemma}\label{lem:join}
	Let $B$ be a join bag such that $B = B_1 = B_2$ for the two children $B_1, B_2$ of $B$.
	Let $\alpha \in X_B$. Let $(\alpha_1^*, \alpha_2^*)$ be optimal solutions to
	\begin{align*}
		\textup{[join]}\qquad\max \ \ & W_{B_1}(\alpha_1) + W_{B_2}(\alpha_2)\\
		\textup{s.t.} \ \ & \alpha_1 + \alpha_2 = \alpha\\
		& \alpha_1 \in X_{B_1},\ \alpha_2 \in X_{B_2}
	\end{align*}
	Then $M_1 \cup M_2 \in \mathcal{M}^*_B(\alpha)$ for all $M_1 \in \mathcal{M}^*_{B_1}(\alpha_1^*),\ M_2 \in \mathcal{M}^*_{B_2}(\alpha_2^*)$.
	If the optimal solution of [join] has value $-\infty$, then $\mathcal{M}^*_B(\alpha) = \emptyset$.
\end{lemma}
\begin{proof}
  Let $M^* := M_1 \cup M_2$ for some $M_1 \in \mathcal{M}^*_{B_1}(\alpha_1^*),\ M_2 \in \mathcal{M}^*_{B_2}(\alpha_2^*)$. We first observe that $V_{B_1} \cap V_{B_2} = B$ by Properties 2 and 3
  of the tree decomposition and hence $M_1 \cap M_2 = \emptyset$. This implies that
  \begin{align*}
   \deg_{M^*}(v) = \begin{cases}
   \deg_{M_1}(v) \in \{0, \ell(v), \dots, u(v)\} & \text{if } v \in V_{B_1} \setminus B,\\
   \deg_{M_2}(v) \in \{0, \ell(v), \dots, u(v)\} & \text{if } v \in V_{B_2} \setminus B,\\
   \deg_{M_1}(v) + \deg_{M_2}(v) = \alpha(v) & \text{if } v \in B.
   \end{cases}
   \end{align*}
  Hence $M^* \in \mathcal{M}_B(\alpha)$.
   
  Now let $M' \in \mathcal{M}_B(\alpha)$. Let $M'_1 := M' \cap E_{B_1}$ and $M'_2 := M' \cap E_{B_2}$. We observe that $(\alpha_{M_1, B_1}, \alpha_{M_2, B_2})$ is a feasible solution to [join] and hence
   $w(M') = w(M'_1) + w(M'_2) \leq w(M_1) + w(M_2) = w(M^*)$.\qed
\end{proof}

Finally, we observe that after computing $W_R(\alpha)$ and the corresponding elements of $\mathcal{M}^*_R(\alpha)$ for each $\alpha$ for the root bag $R$, an optimal assignment for $G$ can be easily obtained.
\begin{lemma}\label{lem:root}
  Let $(S^*, \alpha^*)$ be an optimal solution of
  \begin{align*}
 \textup{[root]}\qquad \max\ \ & W_R(\alpha) + w(S) \\
 \textup{s.t.}\ \ & \alpha(v) + \deg_S(v) \in \{0, \ell(v), \dots, u(v)\} \quad \forall\,v \in R\\
  & \alpha \in X_R,\ S \subseteq E[R].
\end{align*}
  Then $S^* \cup M$ is an optimal solution to \textsc{wmlq} for any $M \in \mathcal{M}^*_R(\alpha^*)$.
\end{lemma}

\begin{proof}
Let $M^* := S^* \cup M$ for some $M \in \mathcal{M}^*_R(\alpha^*)$. Note that for $v \in V \setminus R$, we have $S^* \cap \delta(v) = \emptyset$ and hence $\deg_{M^*}(v) = \deg_M(v) \in \{0, \ell(v), \dots, u(v)\}$ by the feasibility of~$M$. Furthermore, for $v \in R$, we have $\deg_{M^*}(v) = \alpha^*(v) + \deg_{S^*}(v) \in \{0, \ell(v), \dots, u(v)\}$ by the feasibility of $S^*$ for [root]. We conclude that $M^*$ is indeed a feasible solution to \textsc{wmlq}.

Now let $M' \subseteq E$ be any solution to \textsc{wmlq}.
% and let $\alpha' \in \mathbb{Z}^r$ be defined by $\alpha'(v) := \deg_{M'}(v)$ for all $v \in r$. 
Define $S' := M' \cap E[R]$ and $\alpha' := \alpha_{M',R} - \alpha_{S',R}$. Observe that $(S', \alpha')$ is a feasible solution to [root] and that further $M' \setminus S' \in \mathcal{M}_R(\alpha')$.
We conclude that 
\begin{align*}
  w(M') \leq W_R(\alpha') + w(S') \leq W_R(\alpha^*) + w(S^*) = w(M^*),
\end{align*}
and thus $M^*$ is indeed an optimal solution to \textsc{wmlq}.\qed
\end{proof}

\begin{restatable}{theorem}{restateThmFPT}
\label{th:bounded_tw}
\textsc{wmlq} can be solved in time $O(T + (u_{\max})^{3\tw(G)}|E|)$, where $T$ is the time needed for computing a tree decomposition of $G$ of width $\tw(G)$. In particular, {\sc wmlq} can be solved in polynomial time when restricted to instances of bounded treewidth, and {\sc wmlq} parameterized by $\max\{\tw(G), u_{\max}\}$ is fixed-parameter tractable. 
\end{restatable}

\begin{proof} 
  In order to solve a given \textsc{wmlq} instance, the algorithm starts by computing a nice tree decomposition of~$G$ of width $\tw(G)$.  Note that $T$ is of the same order for tree decompositions and nice tree decompositions. Using \cref{lem:leaf,lem:introduce,lem:forget,lem:join,lem:root}, we can inductively compute a representative $M \in \mathcal{M}^*_B(\alpha)$ for every bag $B$ and every $\alpha \in X_b$, or deduce that $\mathcal{M}^*_B(\alpha) = \emptyset$. We first observe that $|X_B| = (u_{\max})^{\tw(G)}$, thus only $(u_{\max})^{\tw(G)}$ representatives have to be computed per bag. Furthermore, for each of the above lemmas, the necessary computations to derive an $M \in \mathcal{M}^*_B(\alpha)$ from representatives of $\mathcal{M}^*_{B'}(\alpha')$ of children $B'$ of $B$ can be done in time $O((u_{\max})^{2\tw(G) + 1})$. This is obvious for \cref{lem:leaf,lem:introduce}. For \cref{lem:forget,lem:join,lem:root} we observe that the sets of feasible solutions for the corresponding optimization problems [forget], [join], and [root] have size at most $2^{|B|} \cdot (u_{\max} + 1)$,  $(u_{\max})^{2\tw(G)}$, and $2^{|R|^2} \cdot (u_{\max})^{\tw(G)}$, respectively (note that without loss of generality we can assume $|R|$ to be of constant size by introducing at most $\tw(G)$ additional forget bags). The theorem then follows from the fact that the number of bags is linear.\qed
\end{proof}

\subsection{$\W[1]$-hardness for parameterizing by treewidth only}

While our algorithm runs in polynomial time for bounded treewidth, the degree of the polynomial depends on the treewidth and the algorithm only becomes a fixed-parameter algorithm when parameterizing by treewidth and $u_{\max}$ simultaneously. We will now show by a reduction from \textsc{Minimum Maximum Outdegree} that this dependence is necessary under the assumption that $\FPT \neq \W[1]$. %We show this by a reduction from \textsc{Minimum Maximum Outdegree}, a graph orientation problem that was shown to be $\W[1]$-hard when parameterized by the treewidth in~\cite{szeider2011not}.

\begin{pr}\textsc{Minimum Maximum Outdegree}\ \\
	Input: A graph $G = (V, E)$, edge weights $w: E \rightarrow \mathbb{Z}_+$ encoded in unary and a degree-bound $r \in \mathbb{Z}_+$.\\
Task: Find an orientation $D$ of $G$ such that $\sum_{e \in \delta_D^+(v)} w(e) \leq r$ for all $v \in V$, where $\delta_D^+(v)$ stands for the set of edges oriented so that their tail is~$v$.
\end{pr}

\begin{theorem}[Theorem~5 from~\cite{Sze11}]
  \textsc{Minimum Maximum Outdegree} is $W[1]$-hard when parameterized by treewidth.
\end{theorem}

\newcommand{\edgestartpost}[2]{p_{#1,#2}\xspace}
\newcommand{\edgeendpost}[2]{p_{#1,#2}\xspace}
\newcommand{\nodepost}[1]{p_{#1}\xspace}
\newcommand{\edgestartapp}[3]{a^{#3}_{#1,#2}\xspace}
\newcommand{\edgeendapp}[3]{a^{#3}_{#1,#2}\xspace}
\newcommand{\edgeapp}[1]{z_{#1}\xspace}

\begin{theorem}
  \textsc{mlq} is $\W[1]$-hard when parameterized by treewidth, even when restricted to instances where $\ell(p) \in \{0, u(p)\}$ for every $p \in P$.
\end{theorem}

\begin{proof}
  Given an instance $(G = (V, E), w, r)$ of \textsc{Minimum Maximum Outdegree}, we construct an instance $(G' = (A\,\dot\cup\, P, E'), \ell, u)$ of \textsc{mlq} as follows:
  \begin{itemize}
   \item For every vertex $v \in V$ we introduce a post $\nodepost{v} \in P$ with lower quota $\ell(\nodepost{v}) = 0$ and upper quota $u(\nodepost{v}) = r$. 
   \item For every edge $e = \{v, v'\} \in E$, we introduce two posts $\edgestartpost{e}{v}$ and $\edgeendpost{e}{v'}$ with identical lower and upper quotas of $w(e)+1$, i.\,e.\,, 
   $$\ell(\edgestartpost{e}{v}) = \ell(\edgeendpost{e}{v'}) = u(\edgestartpost{e}{v}) = u(\edgeendpost{e}{v'}) = w(e) + 1.$$
   We also add $2w(e) + 1$ applicants $ \edgestartapp{e}{v}{1}, \dots,  \edgestartapp{e}{v}{w(e)}, \edgeendapp{e}{v'}{1}, \dots, \edgeendapp{e}{v'}{w(e)}, \edgeapp{e}$, 
   which are connected to the posts by the edges 
   $$\{\nodepost{v}, \edgestartapp{e}{v}{i} \}, \{ \edgestartapp{e}{v}{i}, \edgestartpost{e}{v}\}, \{\nodepost{v'}, \edgeendapp{e}{v'}{i}\}, \{\edgeendapp{e}{v'}{i}, \edgeendpost{e}{v'}\} \text{ for } i \in \{1, \dots, w(e)\}$$ as well as $\{\edgestartpost{e}{v},  \edgeapp{e}\}$ and $\{ \edgeapp{e}, \edgeendpost{e}{v'}\}$. This construction is shown in \cref{fi:theorem6}.
  \end{itemize}

\begin{figure}[htbp]
\centering
\begin{tikzpicture}[scale=0.72, transform shape]

\tikzstyle{post} = [circle, draw=black, fill=black, scale= 0.8]
\tikzstyle{applicant} = [circle, draw=black, fill=white, scale= 0.8]
\tikzstyle{edgelabel} = [rectangle, fill=white]
\definecolor{MyPurple}{RGB}{197,0,205}
\pgfmathsetmacro{\d}{0.7}
\pgfmathsetmacro{\b}{3}

%%%%%%%%%%%%%%%%%%%%%%%%%%%%%%%%%%%%%%%%%%%%%%%%%%%

\begin{scope}[xshift=12cm, yshift=5cm]
 \node[post, label=below:$v$]  (v) at            (0, 0) {};
 \node[post, label=above:$w$]  (w) at            (1.5, 2) {};
 \node[post, label=below:$x$]  (x) at            (3, 0) {};
 \draw[very thick] (v) -- node[auto] {$e$} node[auto,swap] {$2$} (w) -- node[auto] {$f$} node[auto,swap] {$1$} (x) -- node[auto] {$g$} node[auto,swap] {$3$} (v);
 \node[]  (r) at            (0, 2) {$r = 3$};
\end{scope}

\begin{scope}[xshift=2cm, yshift=6.7cm]
 \node[post,rectangle] (pr) at            (0, 0) {};
 \node[anchor=west] (text) at (pr.east) {\begin{minipage}{4.25cm} post vertex $p_v$ with lower quota~0 and upper quota $r$\end{minipage}};
 \node[post]  (pc) at   (0, -1) {};
 \node[anchor=west] (text) at (pc.east) {\begin{minipage}{4.25cm} post vertex $p_{e,v}$ with lower quota and upper quota $w(e)+1$\end{minipage}}; 
 \node[applicant]  (a) at                            (0, -1.75) {};
 \node[anchor=west] (text) at (a.east) {\begin{minipage}{4.25cm} applicant vertex \end{minipage}};  
 \draw[thick] ($(pr.north west)+(-0.2,0.4)$) rectangle ($(a.south -| text.east) + (0.2,-0.2)$);
\end{scope}

\node[post, rectangle, label=45:$\nodepost{v}$]  (v) at            (16, 3) {};
\node[post, rectangle, label=above:$\nodepost{w}$]  (w) at            (6, 3) {};
\node[post, rectangle, label=above:$\nodepost{x}$]  (x) at            (10, 3) {};

\node[post, label=above:$\edgestartpost{e}{v}$] (ev) at    (2, 3) {};
\node[post, label=above:$\edgeendpost{e}{w}$] (ew) at      (4, 3) {};
\node[applicant, label=below:$\edgestartapp{e}{v}{1}$] (ae_2) at (1, 0) {};
\node[applicant, label=below:$\edgestartapp{e}{v}{2}$] (ae_3) at (2, 0) {};
\node[applicant, label=below:$\edgeapp{e}$] (ae_1) at            (3, 0) {};
\node[applicant, label=below:$\edgeendapp{e}{w}{1}$] (ae_4) at   (4, 0) {};
\node[applicant, label=below:$\edgeendapp{e}{w}{2}$] (ae_5) at   (5, 0) {};

\node[post, label=above:$\edgestartpost{f}{w}$] (fw) at    (7.5, 3) {};
\node[post, label=above:$\edgeendpost{f}{x}$] (fx) at      (8.5, 3) {};
\node[applicant, label=below:$\edgestartapp{f}{w}{1}$] (af_2) at (7, 0) {};
\node[applicant, label=below:$\edgeapp{f}$] (af_1) at             (8, 0) {};
\node[applicant, label=below:$\edgeendapp{f}{x}{1}$] (af_3) at   (9, 0) {};

\node[post, label=above:$\edgestartpost{g}{x}$] (gx) at    (13, 3) {};
\node[post, label=above:$\edgeendpost{g}{v}$] (gv) at      (15, 3) {};
\node[applicant, label=below:$\edgestartapp{g}{x}{1}$] (ag_2) at (11, 0) {};
\node[applicant, label=below:$\edgestartapp{g}{x}{2}$] (ag_3) at (12, 0) {};
\node[applicant, label=below:$\edgestartapp{g}{x}{3}$] (ag_4) at (13, 0) {};
\node[applicant, label=below:$\edgeapp{g}$] (ag_1) at            (14, 0) {};
\node[applicant, label=below:$\edgeendapp{g}{v}{1}$] (ag_5) at   (15, 0) {};
\node[applicant, label=below:$\edgeendapp{g}{v}{2}$] (ag_6) at   (16, 0) {};
\node[applicant, label=below:$\edgeendapp{g}{v}{3}$] (ag_7) at   (17, 0) {};

\draw [very thick] (ae_1) --  (ev);
\draw [very thick] (ae_1) --  (ew);
\draw [very thick] (ae_2) --  (ev);
\draw [very thick] (ae_3) --  (ev);
\draw [very thick] (ae_4) --  (ew);
\draw [very thick] (ae_5) --  (ew);

\draw [very thick] (af_1) --  (fw);
\draw [very thick] (af_1) --  (fx);
\draw [very thick] (af_2) --  (fw);
\draw [very thick] (af_3) --  (fx);

\draw [very thick] (ag_1) --  (gx);
\draw [very thick] (ag_1) --  (gv);
\draw [very thick] (ag_2) --  (gx);
\draw [very thick] (ag_3) --  (gx);
\draw [very thick] (ag_4) --  (gx);
\draw [very thick] (ag_5) --  (gv);
\draw [very thick] (ag_6) --  (gv);
\draw [very thick] (ag_7) --  (gv);

\draw [very thick] (af_3) --  (x);
\draw [very thick] (ag_2) --  (x);
\draw [very thick] (ag_3) --  (x);
\draw [very thick] (ag_4) --  (x);
\draw [very thick] (ae_2) -- ++(0,4) -| (v.25);
\draw [very thick] (ae_3) -- ++(1,3.8) -| (v);
\draw [very thick] (ag_5) --  (v);
\draw [very thick] (ag_6) --  (v);
\draw [very thick] (ag_7) --  (v);
\draw [very thick] (af_2) --  (w);
\draw [very thick] (ae_4) --  (w);
\draw [very thick] (ae_5) --  (w);

%\draw [very thick, dotted, MyPurple] (m_2) to[out=120,in=0, distance=1cm ] (q_21);

\end{tikzpicture}
\caption{The transformation of the \textsc{Minimum Maximum Outdegree} instance in the upper right corner to a \textsc{mlq} instance. The numbers on the edges of the \textsc{Minimum Maximum Outdegree} instance are the edge weights.}
\label{fi:theorem6}
\end{figure}
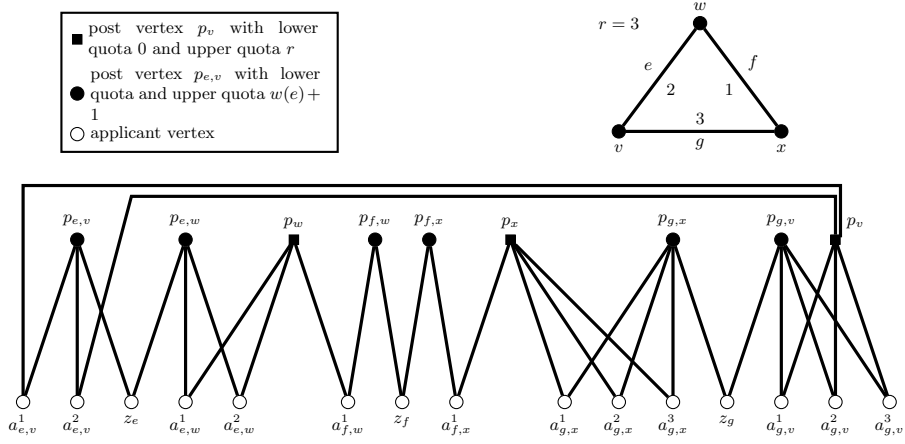

  We show that the constructed instance has a solution serving all applicants if and only if the \textsc{Minimum Maximum Outdegree} instance has an orientation respecting the bound on the weighted outdegree. 
  
  First assume there is an orientation $D$ of $G$ with maximum weighted outdegree at most~$r$. Then consider the assignment that assigns for every oriented edge $(v, v') \in D$ the $w(e)$ applicants $a^i_{e, v}$ to $p_v$ and the $w(e) + 1$ applicants $a^i_{e, v'}$ and $z_e$ to~$p_{e, v'}$. As the weighted outdegree of vertex $v$ is at most $r$, every post $p_{v}$ gets assigned at most $r = u(p_v)$ applicants.
  
  Now assume $M$ is a feasible assignment of applicants to posts serving every applicant. In particular, for every edge $e = \{v,v'\} \in E$, applicant $z_e$ is assigned to either $p_{e, v}$ or $p_{e, v'}$ and exactly one of these two posts is open because the lower bound of $w(e) + 1$ can only be met if $z_e$ is assigned to the respective post. If $p_{e, v}$ is open then all $w(e)$ applicants $a^i_{e, v'}$ are assigned to $p_{v'}$ and none of the applicants $a^i_{e, v}$ is assigned to $p_{v}$, and vice versa if $p_{e, v'}$ is open. Consider the orientation obtained by orienting every edge $e$ from $v$ to $v'$ if and only if $p_{e, v}$ is open. By the above observations, the weighted outdegree of vertex $v$ corresponds to the number of applicants assigned to post $p_v$, which is at most~$r$.
  
  Finally, note that $G'$ can be constructed in time polynomial in the input size of the \textsc{Minimum Maximum Outdegree} instance as the weights are encoded in unary there. Furthermore, the treewidth of $G'$ is at most $\max \{\tw(G), 3\}$. To see this, start with a tree decomposition of $G$ and identify each vertex $v \in V$ with the corresponding post~$p_v$. For every edge $e = \{v, v'\} \in E$, there is a bag $B$ with $p_v, p_v' \in B$. We add the new bag $B_e = \{p_v, p_v', p_{e, v}, p_{e, v'}\}$ as a child to~$B$. We further add the bags $B_{z_e} = \{p_{e, v}, p_{e, v'}, z_e\}$, $B_{a^i_{e,v}} = \{p_v, p_{e, v}, a^i_{e,v}\}$ and $B_{a^i_{e,v'}} = \{p_{v'}, p_{e, v'}, a^i_{e,v'}\}$ for $i \in \{1, \dots, w(e)\}$ as children to~$B_e$. Observe that the tree of bags generated by this construction is a tree decomposition. Furthermore, since we did not increase the size of any of the existing bags and added only bags of size at most $4$, the treewidth of $G'$ is at most $\max \{\tw(G), 3\}$.\qed
\end{proof}

%-------------------------------------------------
\section{Approximation}
\label{se:approx}
Having established the hardness of \textsc{wmlq} even for very restricted instances in Theorem~\ref{th:max_spa_np}, we turn our attention towards approximability. In this section, we give an approximation algorithm and corresponding inapproximability bounds expressed in terms of $|A|, |P|$ and upper quotas in the graph.

The method, which is described formally in Algorithm~\ref{al:greedy}, is a simple greedy algorithm. We say a post $p$ is \emph{admissible} if it is not yet open and $|\Gamma(p)| \geq \ell(p)$. The algorithm iteratively opens an admissible post maximizing the assignable weight, i.e., it finds a post $p' \in P$ and a set $A'$ of applicants in its neighborhood $\Gamma(p')$ with $\ell(p') \leq |A'| \leq u(p')$ such that $\sum_{a \in A'} w(a, p')$ is maximized among all such $(p', A')$ pairs. It then removes the assigned applicants from the graph (potentially rendering some posts inadmissible) and re-iterates until no admissible post is~left.

\begin{algorithm}[h]
\caption{Greedy algorithm for \textsc{wmlq}}
\label{al:greedy}
\begin{algorithmic}
\State Initialize $P_0 = \{p \in P \,:\, |\Gamma(p)| \geq \ell(p)\}$.
\State Initialize $A_0 = A$.
\While{$P_0 \neq \emptyset$}
    \State Find a pair $p' \in P_0$ and $A' \subseteq \Gamma(p') \cap A_0$ with $|A'| \leq u(p')$ such that $\sum_{a \in A'} w(a, p')$ is maximized among all such pairs. 
    \State Open $p'$ and assign all applicants in $A'$ to it.
    \State Remove $p'$ from $P_0$ and remove the elements of $A'$ from $A_0$.
    \For{$p \in P_0$ with $\ell(p) > |\Gamma(p) \cap A_0|$}
    	\State Remove $p$ from $P_0$.
    \EndFor
\EndWhile
\end{algorithmic}
\end{algorithm}

\begin{remark}
As an alternative to Algorithm~\ref{al:greedy}, one could use a reduction from \textsc{wmlq} to the set packing problem. The elements in the universe of the set packing problem would be $A\cup P$. For each post $p$ and for each subset $S \subset \Gamma(p)$, such that $l(p) \le |S| \le u(p)$, we create a set $S \cup \{p\}$ for the set packing instance. A feasible set packing then corresponds to a feasible assignment of the same weight. However, if the difference between $p$'s upper and lower quota is not bounded by a constant, this would create an exponential-sized input for the set packing problem and we could only employ an oracle-based algorithm known for the set packing problem to solve \textsc{wmlq}. The greedy algorithm known for the set packing problem~\cite{CH01} can be made to work in a fashion similar to the algorithm presented above.
\end{remark}

In the following we give a tight analysis of the algorithm, establishing approximation guarantees in terms of the number of posts $|P|$, the number of applicants $|A|$, and the maximum upper quota $u_{\max} := \max_{p \in P} u(p)$ over all posts. We also provide two examples that show that our analysis of the greedy algorithm is tight for each of the described approximation factors. We further show that the approximation ratios given above for \textsc{wmlq} are almost tight from the point of view of complexity theory.

\begin{restatable}{theorem}{restategreedypos}\label{thm:greedy-approximation}
  Algorithm~\ref{al:greedy} is an $\alpha$-approximation algorithm for \textsc{wmlq} with \linebreak \mbox{$\alpha = \min \{|P|,\, |A|,\, u_{\max} + 1\}$}. Furthermore, for \textsc{mlq}, Algorithm~\ref{al:greedy} is a $\sqrt{|A|} + 1$-approximation algorithm. It can be implemented to run in time $O(|E| \log |E|)$.
\end{restatable}

\begin{proof}
Let $p'_1, \dots, p'_{n}$ be the posts chosen by the algorithm and let $A'_1, \dots, A'_{n}$ be the corresponding sets of applicants. Furthermore, consider an optimal solution of weight~$\OPT$, consisting of open posts $p_1, \dots, p_k$ and the corresponding sets of applicants $A_1, \dots, A_k$ assigned to those posts.

We first observe that the first two approximation ratios of $|P|$ and $|A|$ are already achieved by the initial selection of $p'_1$ and $A'_1$ chosen in the first round of the algorithm. For every~\mbox{$i \in \{1, \dots, k\}$}, post $p_i$ is an admissible post in the first iteration of the algorithm. The first iteration's choice of the pair~$(p'_1, A'_1)$ implies $\sum_{a \in A'_1} w(a, p'_1) \geq \sum_{a \in A_i} w(a, p_i) \geq w(a', p_i)$ for every $a' \in A_i$. As the optimal solution opens at most $|P|$ posts and serves at most $|A|$ applicants, we deduce that~$\min \{|P|, |A|\} \cdot \sum_{a \in A'_1} w(a, p'_1) \geq  \OPT$.

We now turn our attention to the remaining approximation guarantees, which are $u_{\max} + 1$ for \textsc{wmlq} and $\sqrt{|A|} + 1$ for \textsc{mlq}. For every $i \in \{1, \dots, k\}$, let $\pi(i)$ denote the first iteration of the algorithm such that $A'_{\pi(i)} \cap A_i \neq \emptyset$ or $p'_{\pi(i)} = p_i$. This is the first iteration in which post $p_i$ is opened or an applicant assigned to it in the optimal solution becomes assigned. Note that such an iteration exists, because $p_i$ is not admissible after the termination of the algorithm. Furthermore, observe that $\sum_{a \in A'_{\pi(i)}} w(a, p'_{\pi(i)}) \geq \sum_{a \in A_i} w(a, p_i)$, because the pair $(p_i, A_i)$ was a valid choice for the algorithm in iteration $\pi(i)$. Now for iteration $j$ define $P_j := \{i \, : \, \pi(i) = j\}$ and observe that $|P_j| \leq |A'_j| + 1$, because $P_j$ can only contain one index~$i'$ with $p_{i'} = p'_j$, and all other $i \in P_j \setminus \{i'\}$ must have $A_i \cap A'_j \neq \emptyset$ (where the sets $A_i$ are disjoint). We conclude that
\begin{align*}
\OPT & \ = \ \sum_{i = 1}^k \sum_{a \in A_i} w(a, p_i) \ \leq \ \sum_{i = 1}^k \sum_{a \in A'_{\pi(i)}} w(a, p'_{\pi(i)}) \\
& \ \leq \ \sum_{j = 1}^{n} |P_j| \sum_{a \in A'_j} w(a, p'_j) \ \leq \ \sum_{j = 1}^{n} (|A'_j| + 1) \sum_{a \in A'_j} w(a, p'_j).
\end{align*}

Note that $|A'_j| \leq u_{\max}$ and therefore \[\OPT \leq (u_{\max} + 1) \sum_{j = 1}^{n} \sum_{a \in A'_j} w(a, p'_j),\] proving the third approximation guarantee. Now consider the unit-weight \textsc{mlq} case and define $A' = \bigcup_{j = 1}^{n} A'_j$. If $|A'| \geq \sqrt{|A|}$, then $\sqrt{|A|}|A'| \geq |A| \geq \OPT$. Therefore assume $|A'| < \sqrt{|A|}$.
Note that in this case, the above inequalities imply $\OPT \leq (|A'| + 1)|A'| \leq (\sqrt{|A|} + 1)|A'|$, proving the improved approximation guarantee for \textsc{mlq}.

We now turn to proving the bound on the running time. We will describe how to implement the search for the greedy choice of the pair $(p', A')$ in each iteration efficiently using a heap data structure. Initially, for every post $p$, we sort the applicants in its neighborhood by non-increasing order of $w(a, p)$. This takes time at most $O(|E| \log |E|)$ as the total number of entries to sort is $\sum_{p \in P} |\Gamma(p)| = |E|$. We then introduce a heap containing all admissible posts, and associate with each post $p$ the total weight of the first $u(p)$ edges in its neighborhood list. 
Note that these entries can be easily kept up to date whenever the algorithm opens a post and assigns applicants to it: In the list of every other post $p$ we simply replace the assigned applicants with the first not-yet-assigned entry in the list (or we remove the post if less than $\ell(p)$ applicants are available). As every edge in the graph can only trigger one such replacement, only $O(|E|)$ updates can occur and each of these requires $O(\log |P|)$ time for reinserting the post at the proper place in the heap. Now, in each iteration of the algorithm, the optimal pair $(p', A')$ can be found by retrieving the maximum element from the heap. This happens at most $|P|$ times and requires $O(\log |P|)$ time in each step. \qed
\end{proof}

\begin{ex}
The following two examples show that our analysis of the greedy algorithm is (asymptotically) tight for each of the described approximation factors.

\begin{itemize}
  \item[(a)] The bounds $|P|$ and $u_{\max} + 1$ are tight, and $\sqrt{|A|} + 1$ is asymptotically tight:
  
  Consider an instance of \textsc{mlq} with $k + 1$ posts $p_0, \dots, p_k$ and $k(k + 1)$ applicants $a_{0,1}, \dots, a_{0,k}, a_{1,1}, \dots, a_{k,k}$. Let $\ell(p_i) = u(p_i) = k$ for $i \in \{0, \dots, k\}$. Each applicant $a_{i,j}$ applies to post $i$, and if $i>0$, additionally to post $0$.
   For the greedy algorithm, opening post $p_0$ and assigning applicants $a_{1,1}, \dots, a_{k,k}$ to it is a valid choice in its first iteration, after which no further posts are admissible. Thus, it only assigns $k$ applicants in total. The optimal solution, however, can assign all $k(k+1)$ applicants by assigning applicants $a_{i,1}, \dots, a_{i,k}$ to $p_i$ for each $i$. Therefore, the greedy algorithm cannot achieve an approximation factor better than $k + 1$ on this family of instances, for which $|P| = k+1$, $\sqrt{|A|} < k+1$, and $u_{\max} = k$.
   
  \item[(b)] The bound $|A|$ is tight:
  
  To see that the approximation ratio of $|A|$ is  tight for \textsc{wmlq} consider the following instance with $k$ posts $p_1, \dots, p_k$ and $k$ applicants $a_1, \dots, a_k$. Let $\ell(p_i) = 0$ and $u(p_i) = k$ for every $i$. Every applicant applies for every post, and $w(a_i, p_i) = 1$ for every $i$ but $w(a_i, p_j) = \varepsilon$ for every $j \neq i$ for some arbitrarily small $\varepsilon > 0$. In its first iteration, the greedy algorithm might choose to open post $p_1$ and assign all applicants to it. This solution accumulates a weight of $1 + (k - 1)\varepsilon$, while the weight of the optimal solution is $k = |A|$.
  \end{itemize}
\end{ex}

\begin{restatable}{theorem}{restategreedyneg}
\label{th:inappr}
	\textsc{mlq} is not approximable within a factor of $|P|^{1-\varepsilon}$ or $\sqrt{|A|}^{1-\varepsilon}$ or~$u_{\max}^{1-\varepsilon}$ for any~$\varepsilon > 0$, unless $\P = \NP$, even when restricting to instances where $\ell(p) = u(p)$ for every $p \in P$ and $|\Gamma(a)| \leq 2$ for every $a \in A$.
\end{restatable}

\begin{proof}
	Once again we use the maximum independent vertex set problem. Given an instance of \textsc{mis} on a graph $G = (V, E)$ with $|V| = n$ and $|E| = m$, we create an \textsc{mlq} instance with $n$ posts $p_1, \dots, p_n$, post $p_i$ corresponding to vertex $v_i$. We also introduce $n^2 - m$ applicants as follows. Initially, we introduce $n$ applicants $a_{i,1}, a_{i,2}, ..., a_{i,n}$ applying for each post $p_i$. Then, for every edge $\{v_i, v_j\} \in E$, we merge the applicants $a_{i,j}$ and $a_{j,i}$, obtaining a single applicant applying for both $p_i$ and $p_j$. Furthermore, we set $\ell(p_j) = u(p_j) = n$ for every post. This construction is shown in \cref{fi:theorem9}.

\begin{figure}[htbp]
 \centering
\begin{tikzpicture}
\tikzstyle{post} = [circle, draw=black, fill=black, scale= 0.8]
\tikzstyle{applicant} = [circle, draw=black, fill=white, scale= 0.8, align=center]
\begin{scope}[xshift=-3cm, yshift=8cm]
 \node[applicant]  (v1) at (0:1) {$v_1$};
 \node[applicant]  (v2) at (72:1) {$v_2$};
 \node[applicant]  (v3) at (144:1) {$v_3$};
 \node[applicant]  (v4) at (216:1) {$v_4$};
 \node[applicant]  (v5) at (288:1) {$v_5$};
 \draw[very thick] (v1) -- (v2) -- (v3) -- (v4)-- (v5) -- (v1);
\end{scope}

\begin{scope}[xshift=0cm, yshift=8.25cm]
 \node[post]  (pc) at   (0, 0) {};
 \node[anchor=west] (text3) at (pc.east) {\begin{minipage}{4.25cm} post vertex $p_{i}$ with lower quota and upper quota $n$\end{minipage}}; 
 \node[applicant]  (a) at                            (0, -0.75) {};
 \node[anchor=west] (text2) at (a.east) {\begin{minipage}{4.25cm} applicant vertex \end{minipage}};  
 \draw[thick] ($(pc.north west)+(-0.2,0.4)$) rectangle ($(a.south -| text.east) + (0.2,-0.2)$);
\end{scope}

\begin{scope}[xshift=0cm, yshift=1.7cm, every node/.style={anchor=center, text depth=.4ex,text height=2ex,text width=1.3em}, scale=0.8]
 \matrix (A) [matrix of math nodes, nodes = {draw,circle}, column 1/.style={nodes={fill=black,text=white,align=center}},row sep=10mm, column sep=6 mm] 
  { p_1 & a_{1,1} & a_{1,2} & a_{1,3} & a_{1,4} & a_{1,5}\\
    p_2 &         & a_{2,2} & a_{2,3} & a_{2,4} & a_{2,5}\\
    p_3 & a_{3,1} &         & a_{3,3} & a_{3,4} & a_{3,5}\\
    p_4 & a_{4,1} & a_{4,2} &         & a_{4,4} & a_{4,5}\\
    p_5 &         & a_{5,2} & a_{5,3} &         & a_{5,5}\\
  };
	 %\node[applicant, rectangle, text height=2ex,text width=1.4em,inner sep=3mm] (a_d) at ($(A-3-1) + (-1.5, 0)$) {$a_d$};
	
	%\draw[thick] (a_d) |- (A-1-1);
	%\draw[thick] (a_d) -- (A-2-1);
	%\draw[thick] (a_d) -- (A-3-1);
	%\draw[thick] (a_d) -- (A-4-1);
	%\draw[thick] (a_d) |- (A-5-1);
  \draw[thick] (A-1-1) -- (A-1-2);
  \draw[thick] (A-1-1.45) -- ++ (0.50,0.4) -| (A-1-3);
  \draw[thick] (A-1-1.60) -- ++ (0.33,0.6) -| (A-1-4);
  \draw[thick] (A-1-1.75) -- ++ (0.16,0.8) -| (A-1-5);
  \draw[thick] (A-1-1.90) -- ++ (0.00,1.0) -| (A-1-6);
  \draw[thick] (A-2-1) -- (A-2-1 -| A-1-2) -- (A-1-3);
  \draw[thick] (A-2-1.45) -- ++ (0.50,0.4) -| (A-2-3);
  \draw[thick] (A-2-1.60) -- ++ (0.33,0.6) -| (A-2-4);
  \draw[thick] (A-2-1.75) -- ++ (0.16,0.8) -| (A-2-5);
  \draw[thick] (A-2-1.90) -- ++ (0.00,1) -| (A-2-6);
  \draw[thick] (A-3-1) -- (A-3-2); \coordinate (X) at ($(A-3-1.45)+(0.5,0.4)$); \coordinate (Y) at (X -| A-2-3);
  \draw[thick] (A-3-1.45) -- ++ (0.50,0.4) -- (Y) -- (A-2-4);
  \draw[thick] (A-3-1.60) -- ++ (0.33,0.6) -| (A-3-4);
  \draw[thick] (A-3-1.75) -- ++ (0.16,0.8) -| (A-3-5);
  \draw[thick] (A-3-1.90) -- ++ (0.00,1) -| (A-3-6);  
  \draw[thick] (A-4-1) -- (A-4-2); \coordinate (X) at ($(A-4-1.60)+(0.33,0.6)$); \coordinate (Y) at (X -| A-3-4);
  \draw[thick] (A-4-1.45) -- ++ (0.50,0.4) -| (A-4-3);
  \draw[thick] (A-4-1.60) -- ++ (0.33,0.6) -- (Y) -- (A-3-5);
  \draw[thick] (A-4-1.75) -- ++ (0.16,0.8) -| (A-4-5);
  \draw[thick] (A-4-1.90) -- ++ (0.00,1) -| (A-4-6);  
  \draw[thick] (A-5-1) -- (A-5-1 -| A-1-2) -- ++(0,-0.7) -- ++(8.5,0) |- (A-1-6);
  \coordinate (X) at ($(A-5-1.75)+(0.16,0.8)$); \coordinate (Y) at (X -| A-4-5);
  \draw[thick] (A-5-1.45) -- ++ (0.50,0.4) -| (A-5-3);
  \draw[thick] (A-5-1.60) -- ++ (0.33,0.6) -| (A-5-4);
  \draw[thick] (A-5-1.75) -- ++ (0.16,0.8) -- (Y) -- (A-4-6);
  \draw[thick] (A-5-1.90) -- ++ (0.00,1) -| (A-5-6);  
  \end{scope}
\end{tikzpicture}
\caption{The transformation of the \textsc{mis} instance in the upper left corner to a \textsc{mlq} instance.}
\label{fi:theorem9}
\end{figure}
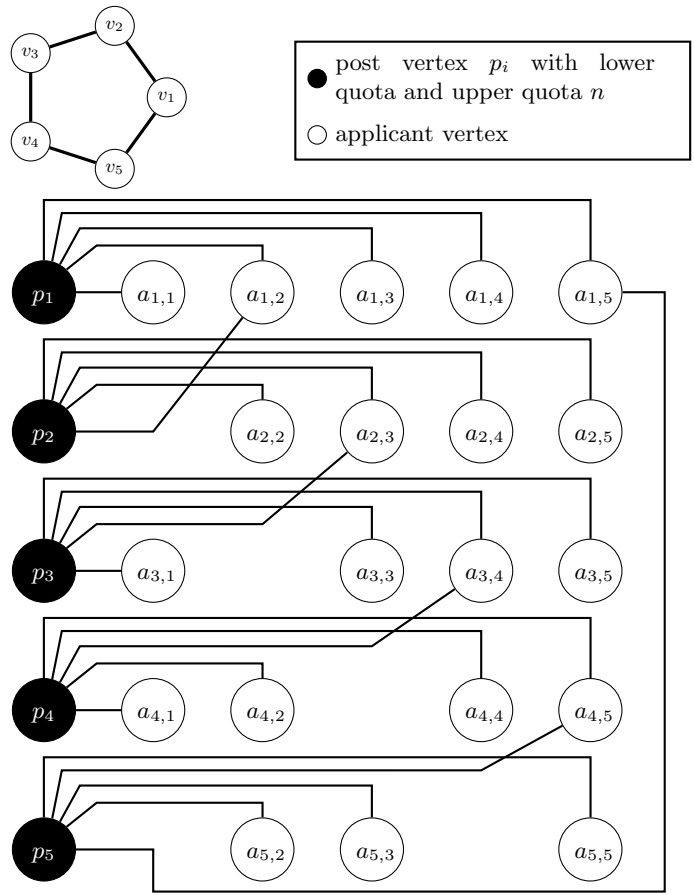

	Note that due to the choice of upper and lower bounds, any open post must be assigned to all the applicants in its neighborhood. Thus, a solution to the \textsc{wmlq} instance is feasible if and only if $\Gamma(p_i) \cap \Gamma(p_j) = \emptyset$ for all open posts $p_i$ and $p_j$ with $i \neq j$, which is equivalent to $v_i$ and $v_j$ not being adjacent in $G$ by construction of the instance. Therefore, the \textsc{mlq} instance has a feasible solution opening $k$ posts (and thus serving $kn$ applicants) if and only if there is an independent set of size $k$ in~$G$. We conclude that $\OPT_{\textsc{mlq}} = n \cdot \OPT_{\textsc{mis}}$ for the two instances under consideration.
	
	Note that in the constructed \textsc{mlq} instance, $n =|P| = u_{\max} \geq \sqrt{|A|}$. Therefore any approximation algorithm with a factor better than $|P|^{1-\varepsilon}$ or $\sqrt{|A|}^{1-\varepsilon}$ or $u_{\max}^{1-\varepsilon}$ for $\varepsilon > 0$ yields a solution of the instance that serves at least~$(1/n^{1-\varepsilon}) \OPT_{\textsc{mlq}}$ applicants and therefore opens at least $(1/n^{2-\varepsilon}) \OPT_{\textsc{mlq}} = (1/n^{1 - \varepsilon}) \OPT_{\textsc{mis}}$ posts, corresponding to an independent set of the same size. By~\cite{Zuc07}, this implies $\P = \NP$. \qed
\end{proof}

\section{Matchings with lower quotas in general graphs}
\label{sec:many}

Throughout this paper, we focused on many-to-one matchings in bipartite graphs because these fit most applications in the centralized formation of groups that motivated our investigation. A straightforward generalization of {\sc wmlq} to matchings in an arbitrary (not necessarily bipartite) graph $G$ allows \textit{all} vertices of the graph to have lower and upper quotas. 

\begin{pr} \textsc{gwmlq}\ \\
	Input: $\mathcal{I} = (G, w, \ell, u)$; a not necessarily bipartite graph $G = (V, E)$ with edge weights $w$, lower quotas $\ell$ and upper quotas $u$.\\
Task: Find an assignment of maximum weight.\\
If $w(e)=1$ for all $e \in E$, we refer to the problem as \textsc{gmlq}.
\end{pr}

One can see this generalization as a variant of the $D$-matching problem (see \cref{sec:rWork}), where each vertex has a domain consisting of 0 and an interval. Clearly, the hardness results derived in the previous sections are valid for {\sc gwmlq} as well. We now briefly argue that the positive results from Sections~\ref{sec:com_rest} and \ref{sec:bounded_tw} carry over to this generalized setting. However, our approximation results do not hold even if $G$ is bipartite and only a single applicant is equipped with lower and upper quotas. In fact, {\sc gwmlq} does not allow for any approximation even in this very restricted case unless $\P=\NP$.
\smallskip

The two positive results in \cref{sec:com_rest}, namely Theorems~\ref{th:infty_2} and \ref{th:u2}, are applicable to {\sc gwmlq}. Note that \cref{th:infty_2} (bounded degree for all posts) is a special case of \cref{th:u2} (bounded upper quota for all posts).

\begin{theorem}
	 {\sc gwmlq} can be solved in polynomial time when restricted to instances with $u(v) \leq 2$ for all $v \in V$.
\end{theorem}

\begin{proof}
We will work with the proof of \cref{th:u2}, which requires some simple modifications to fit the case of arbitrary graphs. All we need to do is to add a dummy vertex $v_d$ to $G$ -- this resembles dummy post $p_d$ in the proof of \cref{th:u2}. The steps corresponding to a post vertex should now be executed for all vertices of the graph. We can assume there are no vertices with lower quota~0 and upper quota~2 by a similar reasoning given in \cref{th:infty_2}. For every vertex $v_i$ with $ \ell(v_i) = 2$, we add two dummy vertices $q_i^1$ and $q_i^2$ and connect them to each other and~$v_i$.  Then, the dummy vertex $v_d$ is connected to vertices with upper quota~1. We finish the construction by adding triangles to $v_d$ to ensure that only two $f$-factors need to be computed. The arguments in the proof of \cref{th:u2} can now be applied to this $f$-factor instance. \qed
\end{proof}

As for \cref{th:bounded_tw}, the algorithm for bounded treewidth and upper quota carries over to {\sc gwmlq} without any modification. Note that in the proof we never used the bipartiteness of $G$ or that $u(a) =1$ for the applicants.

\begin{theorem}
	  {\sc gwmlq} can be solved in time $O(T + (u_{\max})^{3\tw(G)}|E|)$, where $T$ is the time needed for computing a tree decomposition of $G$ of width $\tw(G)$. In particular, {\sc gwmlq} can be solved in polynomial time when restricted to instances of bounded treewidth, and {\sc wmlq} parameterized by $\max\{\tw(G), u_{\max}\}$ is fixed-parameter tractable.
\end{theorem}

	Finally, we prove that \cref{al:greedy} cannot be generalized even for bipartite {\sc mlq} with lower and upper quotas on both sides.

\begin{theorem}
	It is $\NP$-hard to decide whether $\OPT > 0$ for an instance of {\sc gmlq}, even if the graph is bipartite and on one side of the bipartition all vertices except for one have unitary upper and lower quota.
\end{theorem}

\begin{proof}
	To every instance of {\sc mis} we construct an instance of {\sc gmlq} so that the {\sc mis} instance admits an independent set of size $K$ if and only if $\OPT > 0$ for the {\sc gmlq} instance. We start with the same {\sc mlq} instance that was constructed from an {\sc mis} instance in the proof of \cref{th:inappr}. The changes are depicted in \cref{fi:theorem11}. A dummy applicant $a_d$ is added to the graph and connected to all posts. We set $\ell(a_d) = u(a_d) = K$ and change $\ell(p) = u(p)$ to $n + 1$ for every post $p \in P$.
	
		\begin{figure}[htbp]
 \centering
\begin{tikzpicture}
\tikzstyle{post} = [circle, draw=black, fill=black, scale= 0.7]
\tikzstyle{applicant} = [circle, draw=black, fill=white, scale= 0.7, align=center]
\begin{scope}[xshift=-3cm, yshift=8cm]
 \node[applicant]  (v1) at (0:1) {$v_1$};
 \node[applicant]  (v2) at (72:1) {$v_2$};
 \node[applicant]  (v3) at (144:1) {$v_3$};
 \node[applicant]  (v4) at (216:1) {$v_4$};
 \node[applicant]  (v5) at (288:1) {$v_5$};
 \draw[very thick] (v1) -- (v2) -- (v3) -- (v4)-- (v5) -- (v1);
\end{scope}

\begin{scope}[xshift=0cm, yshift=8.65cm]
 \node[post]  (pc) at   (0, 0) {};
 \node[anchor=west] (text3) at (pc.east) {\begin{minipage}{4.25cm} post vertex $p_{i}$ with lower quota and upper quota $n + 1$\end{minipage}}; 
 \node[applicant]  (a) at                            (0, -0.75) {};
 \node[anchor=west] (text2) at (a.east) {\begin{minipage}{4.25cm} applicant vertex \end{minipage}};  
	\node[applicant,rectangle]  (a) at                            (0, -1.5) {};
 \node[anchor=west] (text1) at (a.east) {\begin{minipage}{4.25cm} dummy applicant vertex with lower quota and upper quota $K$\end{minipage}};  
 \draw[thick] ($(pc.north west)+(-0.2,0.4)$) rectangle ($(a.south -| text.east) + (0.2,-0.4)$);
\end{scope}

\begin{scope}[xshift=0cm, yshift=1.2cm, every node/.style={anchor=center, text depth=.4ex,text height=2ex,text width=1.3em}, scale=0.8]
 \matrix (A) [matrix of math nodes, nodes = {draw,circle}, column 1/.style={nodes={fill=black,text=white,align=center}},row sep=10mm, column sep=6 mm] 
  { p_1 & a_{1,1} & a_{1,2} & a_{1,3} & a_{1,4} & a_{1,5}\\
    p_2 &         & a_{2,2} & a_{2,3} & a_{2,4} & a_{2,5}\\
    p_3 & a_{3,1} &         & a_{3,3} & a_{3,4} & a_{3,5}\\
    p_4 & a_{4,1} & a_{4,2} &         & a_{4,4} & a_{4,5}\\
    p_5 &         & a_{5,2} & a_{5,3} &         & a_{5,5}\\
  };
	 \node[applicant, rectangle, text height=2ex,text width=1.4em,inner sep=3mm] (a_d) at ($(A-3-1) + (-1.5, 0)$) {$a_d$};
	
	\draw[thick] (a_d) |- (A-1-1);
	\draw[thick] (a_d) -- (A-2-1);
	\draw[thick] (a_d) -- (A-3-1);
	\draw[thick] (a_d) -- (A-4-1);
	\draw[thick] (a_d) |- (A-5-1);
  \draw[thick] (A-1-1) -- (A-1-2);
  \draw[thick] (A-1-1.45) -- ++ (0.50,0.4) -| (A-1-3);
  \draw[thick] (A-1-1.60) -- ++ (0.33,0.6) -| (A-1-4);
  \draw[thick] (A-1-1.75) -- ++ (0.16,0.8) -| (A-1-5);
  \draw[thick] (A-1-1.90) -- ++ (0.00,1.0) -| (A-1-6);
  \draw[thick] (A-2-1) -- (A-2-1 -| A-1-2) -- (A-1-3);
  \draw[thick] (A-2-1.45) -- ++ (0.50,0.4) -| (A-2-3);
  \draw[thick] (A-2-1.60) -- ++ (0.33,0.6) -| (A-2-4);
  \draw[thick] (A-2-1.75) -- ++ (0.16,0.8) -| (A-2-5);
  \draw[thick] (A-2-1.90) -- ++ (0.00,1) -| (A-2-6);
  \draw[thick] (A-3-1) -- (A-3-2); \coordinate (X) at ($(A-3-1.45)+(0.5,0.4)$); \coordinate (Y) at (X -| A-2-3);
  \draw[thick] (A-3-1.45) -- ++ (0.50,0.4) -- (Y) -- (A-2-4);
  \draw[thick] (A-3-1.60) -- ++ (0.33,0.6) -| (A-3-4);
  \draw[thick] (A-3-1.75) -- ++ (0.16,0.8) -| (A-3-5);
  \draw[thick] (A-3-1.90) -- ++ (0.00,1) -| (A-3-6);  
  \draw[thick] (A-4-1) -- (A-4-2); \coordinate (X) at ($(A-4-1.60)+(0.33,0.6)$); \coordinate (Y) at (X -| A-3-4);
  \draw[thick] (A-4-1.45) -- ++ (0.50,0.4) -| (A-4-3);
  \draw[thick] (A-4-1.60) -- ++ (0.33,0.6) -- (Y) -- (A-3-5);
  \draw[thick] (A-4-1.75) -- ++ (0.16,0.8) -| (A-4-5);
  \draw[thick] (A-4-1.90) -- ++ (0.00,1) -| (A-4-6);  
  \draw[thick] (A-5-1) -- (A-5-1 -| A-1-2) -- ++(0,-0.7) -- ++(8.5,0) |- (A-1-6);
  \coordinate (X) at ($(A-5-1.75)+(0.16,0.8)$); \coordinate (Y) at (X -| A-4-5);
  \draw[thick] (A-5-1.45) -- ++ (0.50,0.4) -| (A-5-3);
  \draw[thick] (A-5-1.60) -- ++ (0.33,0.6) -| (A-5-4);
  \draw[thick] (A-5-1.75) -- ++ (0.16,0.8) -- (Y) -- (A-4-6);
  \draw[thick] (A-5-1.90) -- ++ (0.00,1) -| (A-5-6);  
  \end{scope}
\end{tikzpicture}
\caption{The transformation of the \textsc{mis} instance in the upper left corner to a generalized \textsc{mlq} instance.}
\label{fi:theorem11}
\end{figure}
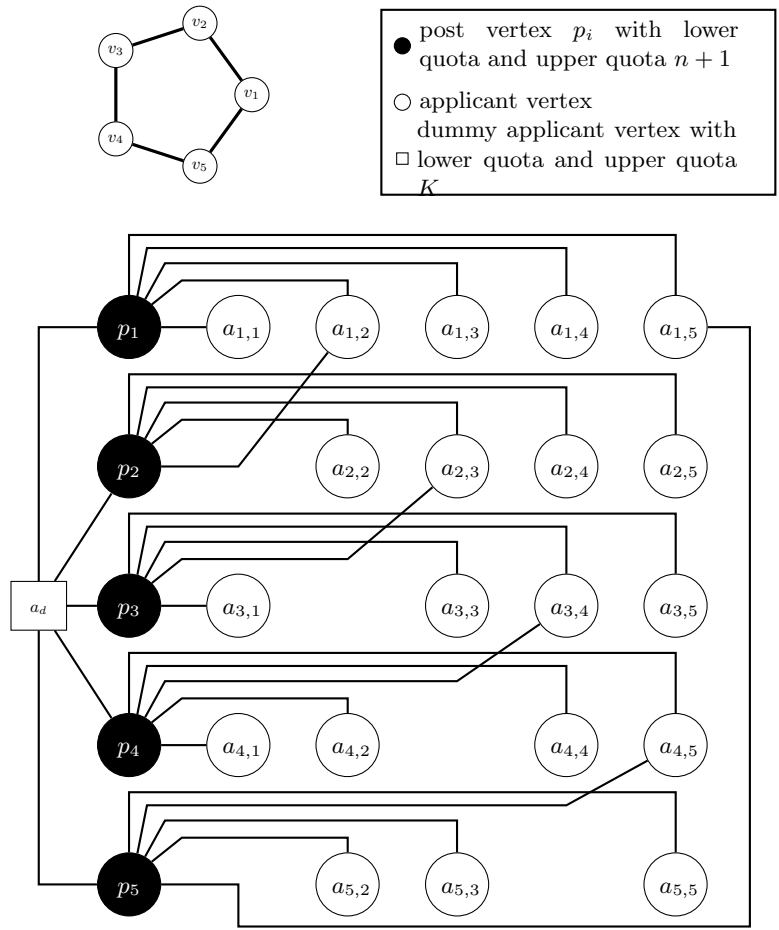
    
    Since every post is adjacent to exactly $n + 1$ applicants, opening a post requires allocating all its applicants to it, including $a_d$ as well. Thus, opening any post implies allocating $a_d$ to exactly $K$ posts. These $K$ open posts do not share applicants other than $a_d$, which is equivalent to the $K$ vertices corresponding to them in the {\sc mis} instance forming an independent set.\qed
\end{proof}

\section{Conclusion}

We discussed the complexity, approximability and fixed-parameter tractability of {\sc wmlq} from various viewpoints such as bounded degree, quota and treewidth.%\accom{I commented out the Tutte set part until Corinna answers my email. I didn't find other papers using this parameter. It also sounds somewhat contrived, but I have no better idea. Maybe we shouldn't suggest the reader to further parameterize the problem.} % Other tractable cases can possibly be discovered if parameterized by other characteristics of the instance, such as the size of the Tutte set in an Edmonds-Gallai decomposition, as recently used in~\cite{gottschalk2016additive}.

Further work on the topic might %also
 include imposing common quotas on some groups of posts. That is, we may have subsets $P_1,\dots,P_k$, where for each $i$ ($1\leq i\leq k$), $P_i\subseteq P$, $P_i$ has a \emph{common quota} $u(P_i)\geq 1$, where $u(P_i)\leq \sum_{p\in P_i}u(p)$, and any assignment $M$ must now satisfy the additional property that $\sum_{p\in P_i} |\delta(p)\cap M|\leq u(P_i)$. Common quotas can model constraints such as the limited availability of resources required for certain projects -- for example $P_1$ might correspond to those projects that require access to high-performance computing facilities.

We have seen that {\sc wmlq} as defined in \cref{pr:wmlq} has a natural application in the context of student-project allocation, where the weight on a given edge $(s,p)$ corresponds to the utility of student $s$ being assigned to project~$p$. However in many applications students have ordinal preferences over projects. Cardinal utilities can of course follow from these via the use of Borda scores, so we can obtain {\sc wmlq} as before.  But ordinal preferences themselves allow alternative optimality criteria to be formulated. For example we may optimize on the \emph{profile} of a matching $M$, which is a vector whose $i$th position indicates the number of students who obtain their $i$th-choice project in $M$~\cite{Man13}. A \emph{greedy maximum matching} is a matching whose profile is lexicographically maximum, taken over all maximum cardinality matchings, whilst a \emph{generous maximum matching} is a matching whose reverse profile is lexicographically minimum, taken over all maximum cardinality matchings. There are efficient algorithms to find greedy and generous maximum matchings in the absence of lower quotas~\cite{KIMS15}, but it remains open to extend the positive results in this paper to the setting involving both lower quotas and preferences.

\subsubsection*{Acknowledgements}
We would like to thank Andr\'as Frank and Krist\'of B\'erczi for their observations that led us to \cref{th:u2}.

\bibliographystyle{plain} %*****bibstyle is supposed to be spmpsci. I'm out of ideas how to make that work with natbib.
\bibliography{mybib}

\begin{thebibliography}{10}

\bibitem{AIM07}
D.~J. Abraham, R.~W. Irving, and D.~F. Manlove.
\newblock Two algorithms for the {S}tudent-{P}roject {A}llocation problem.
\newblock {\em Journal of Discrete Algorithms}, 5(1):79--91, 2007.

\bibitem{AK00}
P.~Alimonti and V.~Kann.
\newblock Some {A}{P}{X}-completeness results for cubic graphs.
\newblock {\em Theoretical Computer Science}, 237(1-2):123--134, 2000.

\bibitem{BFIM10}
P.~Bir\'o, T.~Fleiner, R.~W. Irving, and D.~F. Manlove.
\newblock The {C}ollege {A}dmissions problem with lower and common quotas.
\newblock {\em Theoretical Computer Science}, 411:3136--3153, 2010.

\bibitem{Bod96}
H.~L. Bodlaender.
\newblock A linear-time algorithm for finding tree-decompositions of small
  treewidth.
\newblock {\em SIAM Journal on Computing}, 25(6):1305--1317, 1996.

\bibitem{CH01}
B.~Chandra and M.~M. Halld\'orsson.
\newblock Greedy local improvement and weighted set packing approximation.
\newblock {\em Journal of Algorithms}, 39(2):223--240, 2001.

\bibitem{Cor88}
G.~Cornu\'ejols.
\newblock General factors of graphs.
\newblock {\em Journal of Combinatorial Theory, Series B}, 45(2):185--198,
  1988.

\bibitem{FITUY15}
D.~Fragiadakis, A.~Iwasaki, P.~Troyan, S.~Ueda, and M.~Yokoo.
\newblock Strategyproof matching with minimum quotas.
\newblock {\em ACM Transactions on Economics and Computation}, 4(1):6:1--40,
  Jan 2016.

\bibitem{Gab83}
H.~N. Gabow.
\newblock An efficient reduction technique for degree-constrained subgraph and
  bidirected network flow problems.
\newblock In {\em Proceedings of STOC '83: the 15th Annual ACM Symposium on
  Theory of Computing}, pages 448--456. ACM, 1983.

\bibitem{Gab90}
H.~N. Gabow.
\newblock Data structures for weighted matching and nearest common ancestors
  with linking.
\newblock In {\em Proceedings of SODA '90: the 1st ACM-SIAM Symposium on
  Discrete Algorithms}, pages 434--443. ACM-SIAM, 1990.

\bibitem{GHIKUYY14}
M.~Goto, N.~Hashimoto, A.~Iwasaki, Y.~Kawasaki, S.~Ueda, Y.~Yasuda, and
  M.~Yokoo.
\newblock Strategy-proof matching with regional minimum quotas.
\newblock In {\em Proceedings of the AAMAS '14: the 13th International
  Conference on Autonomous Agents and Multi-agent Systems}, pages 1225--1232.
  IFAAMAS, 2014.

\bibitem{GKHIY15}
M.~Goto, R.~Kurata, N.~Hamada, A.~Iwasaki, and M.~Yokoo.
\newblock Improving fairness in nonwasteful matching with hierarchical regional
  minimum quotas.
\newblock In {\em Proceedings of AAMAS '15: the 14th International Conference
  on Autonomous Agents and Multiagent Systems}, pages 1887--1888. IFAAMAS,
  2015.

\bibitem{HIM14}
K.~Hamada, K.~Iwama, and S.~Miyazaki.
\newblock The hospitals/residents problem with lower quotas.
\newblock {\em Algorithmica}, 74(1):440--465, 2014.

\bibitem{IMY12}
K.~Iwama, S.~Miyazaki, and H.~Yanagisawa.
\newblock Improved approximation bounds for the student-project allocation
  problem with preferences over projects.
\newblock {\em Journal of Discrete Algorithms}, 13:59--66, 2012.

\bibitem{Kam13}
N.~Kamiyama.
\newblock A note on the serial dictatorship with project closures.
\newblock {\em Operations Research Letters}, 41:559--561, 2013.

\bibitem{Klo94}
T.~Kloks.
\newblock {\em Treewidth, Computations and Approximations}, volume 842 of {\em
  Lecture Notes in Computer Science}.
\newblock Springer, 1994.

\bibitem{KIMS15}
A.~Kwanashie, R.W. Irving, D.F. Manlove, and C.T.S. Sng.
\newblock Profile-based optimal matchings in the {S}tudent/{P}roject
  {A}llocation problem.
\newblock In {\em Proceedings of IWOCA '14: the 25th International Workshop on
  Combinatorial Algorithms}, volume 8986 of {\em Lecture Notes in Computer
  Science}, pages 213--225. Springer, 2015.

\bibitem{Lov72}
L.~Lov\'asz.
\newblock On the structure of factorizable graphs.
\newblock {\em Acta Mathematica Academiae Scientiarum Hungaricae},
  23(1-2):179--195, 1972.

\bibitem{Lov73}
L.~Lov\'asz.
\newblock Antifactors of graphs.
\newblock {\em Periodica Mathematica Hungarica}, 4(2-3):121--123, 1973.

\bibitem{Man13}
D.~F. Manlove.
\newblock {\em Algorithmics of Matching Under Preferences}.
\newblock World Scientific, 2013.

\bibitem{MO08}
D.~F. Manlove and G.~O'Malley.
\newblock Student project allocation with preferences over projects.
\newblock {\em Journal of Discrete Algorithms}, 6:553--560, 2008.

\bibitem{MT13}
D.~Monte and N.~Tumennasan.
\newblock Matching with quorums.
\newblock {\em Economics Letters}, 120:14--17, 2013.

\bibitem{Nie06}
R.~Niedermeier.
\newblock {\em Invitation to Fixed-Parameter Algorithms}.
\newblock Oxford University Press, 2006.

\bibitem{Plu07}
M.~D. Plummer.
\newblock Graph factors and factorization: 1985--2003: A survey.
\newblock {\em Discrete Mathematics}, 307(7-8):791--821, 2007.

\bibitem{SS11}
M.~Samer and S.~Szeider.
\newblock Tractable cases of the extended global cardinality constraint.
\newblock {\em Constraints}, 16:1--24, 2011.

\bibitem{Seb93}
A.~Seb\H{o}.
\newblock General antifactors of graphs.
\newblock {\em Journal of Combinatorial Theory, Series B}, 58(2):174--184,
  1993.

\bibitem{Sze11}
S.~Szeider.
\newblock Not so easy problems for tree decomposable graphs.
\newblock {\em CoRR}, abs/1107.1177, 2011.

\bibitem{Zuc07}
D.~Zuckerman.
\newblock Linear degree extractors and the inapproximability of max clique and
  chromatic number.
\newblock {\em Theory of Computing}, 3(6):103--128, 2007.

\end{thebibliography}
\end{document}